\PassOptionsToPackage{unicode}{hyperref}
\PassOptionsToPackage{hyphens}{url}
\PassOptionsToPackage{dvipsnames,svgnames,x11names}{xcolor}
\documentclass[
  11pt,letterpaper]{article}
\usepackage{xcolor}
\usepackage{amsmath,amssymb}
\usepackage{amsthm}
\usepackage{setspace}
\newtheorem{theorem}{Theorem}[section]

\newtheorem{proposition}{Proposition}[section]
\newtheorem{corollary}{Corollary}[theorem]
\newtheorem{lemma}[theorem]{Lemma}
\newtheorem{assumption}{Assumption}[section]
\theoremstyle{remark}
\newtheorem{remark}{Remark}
\DeclareMathOperator*{\argmin}{arg\,min}

\usepackage{algorithm}
\usepackage{algpseudocode}

\usepackage{iftex}
\ifPDFTeX
  \usepackage[T1]{fontenc}
  \usepackage[utf8]{inputenc}
  \usepackage{textcomp} 
\else 
  \usepackage{unicode-math} 
  \defaultfontfeatures{Scale=MatchLowercase}
  \defaultfontfeatures[\rmfamily]{Ligatures=TeX,Scale=1}
\fi
\usepackage{lmodern}
\ifPDFTeX\else
\fi
\IfFileExists{upquote.sty}{\usepackage{upquote}}{}
\IfFileExists{microtype.sty}{
  \usepackage[]{microtype}
  \UseMicrotypeSet[protrusion]{basicmath} 
}{}
\makeatletter
\@ifundefined{KOMAClassName}{
  \IfFileExists{parskip.sty}{%
    \usepackage{parskip}
  }{
    \setlength{\parindent}{0pt}
    \setlength{\parskip}{6pt plus 2pt minus 1pt}}
}{
  \KOMAoptions{parskip=half}}
\makeatother
\makeatletter
\ifx\paragraph\undefined\else
  \let\oldparagraph\paragraph
  \renewcommand{\paragraph}{
    \@ifstar
      \xxxParagraphStar
      \xxxParagraphNoStar
  }
  \newcommand{\xxxParagraphStar}[1]{\oldparagraph*{#1}\mbox{}}
  \newcommand{\xxxParagraphNoStar}[1]{\oldparagraph{#1}\mbox{}}
\fi
\ifx\subparagraph\undefined\else
  \let\oldsubparagraph\subparagraph
  \renewcommand{\subparagraph}{
    \@ifstar
      \xxxSubParagraphStar
      \xxxSubParagraphNoStar
  }
  \newcommand{\xxxSubParagraphStar}[1]{\oldsubparagraph*{#1}\mbox{}}
  \newcommand{\xxxSubParagraphNoStar}[1]{\oldsubparagraph{#1}\mbox{}}
\fi
\makeatother

\usepackage{longtable,booktabs,array,multirow}
\usepackage{calc} 
\usepackage{etoolbox}
\makeatletter
\patchcmd\longtable{\par}{\if@noskipsec\mbox{}\fi\par}{}{}
\makeatother
\IfFileExists{footnotehyper.sty}{\usepackage{footnotehyper}}{\usepackage{footnote}}
\makesavenoteenv{longtable}
\usepackage{graphicx}
\usepackage{placeins}
\usepackage{tikz}
\usetikzlibrary{arrows.meta, positioning, shapes.geometric, calc}
\makeatletter
\newsavebox\pandoc@box
\newcommand*\pandocbounded[1]{
  \sbox\pandoc@box{#1}%
  \Gscale@div\@tempa{\textheight}{\dimexpr\ht\pandoc@box+\dp\pandoc@box\relax}%
  \Gscale@div\@tempb{\linewidth}{\wd\pandoc@box}%
  \ifdim\@tempb\p@<\@tempa\p@\let\@tempa\@tempb\fi
  \ifdim\@tempa\p@<\p@\scalebox{\@tempa}{\usebox\pandoc@box}%
  \else\usebox{\pandoc@box}%
  \fi%
}
\def\fps@figure{htbp}
\makeatother

\usepackage[authoryear,round]{natbib}
\usepackage[margin=1in]{geometry}
\makeatletter
\@ifpackageloaded{caption}{}{\usepackage{caption}}
\AtBeginDocument{%
\ifdefined\contentsname
  \renewcommand*\contentsname{Table of contents}
\else
  \newcommand\contentsname{Table of contents}
\fi
\ifdefined\listfigurename
  \renewcommand*\listfigurename{List of Figures}
\else
  \newcommand\listfigurename{List of Figures}
\fi
\ifdefined\listtablename
  \renewcommand*\listtablename{List of Tables}
\else
  \newcommand\listtablename{List of Tables}
\fi
\ifdefined\figurename
  \renewcommand*\figurename{Figure}
\else
  \newcommand\figurename{Figure}
\fi
\ifdefined\tablename
  \renewcommand*\tablename{Table}
\else
  \newcommand\tablename{Table}
\fi
}
\@ifpackageloaded{float}{}{\usepackage{float}}
\floatstyle{ruled}
\@ifundefined{c@chapter}{\newfloat{codelisting}{h}{lop}}{\newfloat{codelisting}{h}{lop}[chapter]}
\floatname{codelisting}{Listing}

\makeatother
\makeatletter
\@ifpackageloaded{caption}{}{\usepackage{caption}}
\@ifpackageloaded{subcaption}{}{\usepackage{subcaption}}
\makeatother
\usepackage{bookmark}
\IfFileExists{xurl.sty}{\usepackage{xurl}}{} 
\hypersetup{
  pdftitle={Adaptive Probe Design for Auditing AI Agents},
  pdfauthor={Zhenghan Li, Xuan Bi, Jie Ding},
  pdfkeywords={AI safety, AI agent evaluation, offline policy evaluation, importance sampling, reinforcement learning},
  colorlinks=true,
  linkcolor={black},
  filecolor={black},
  citecolor={black},
  urlcolor={black},
  pdfcreator={LaTeX}}

\newcommand{\Var}{\mathrm{Var}}

\begin{document}

\def\spacingset#1{\renewcommand{\baselinestretch}%
{#1}\small\normalsize} \spacingset{1}


\date{}
\title{\bf Adaptive Probe Design for Auditing AI Agents}
\author{
Zhenghan Li\textsuperscript{1} \qquad
Xuan Bi\textsuperscript{2} \qquad
Jie Ding\textsuperscript{1}\\[0.5em]
\small \textsuperscript{1}School of Statistics, University of Minnesota\\
\small \textsuperscript{2}Carlson School of Management, University of Minnesota
}
\maketitle

\begin{abstract}
How can we audit an AI agent before it is safe and practical to deploy? A central task is to estimate its performance under the intended deployment policy (the target policy). However, historical observations may provide insufficient coverage of relevant states and actions, limiting evaluation reliability. We propose CAD (Flow-Constrained Audit Design), a framework that combines historical logs with a limited budget of targeted audit trials (probes). The key idea is to design probes that minimize the mismatch between the combined state–action visitation distribution (occupancy) of the logs and probes and that of the target policy. We formulate this design as a constrained convex optimization problem, directing probes toward regions important to the target policy but underrepresented in historical logs. We develop an adaptive algorithm that iteratively estimates environment dynamics and occupancies to update the probe policy. We link occupancy mismatch to bounds on mean squared error and asymptotic variance, and establish error bounds for CAD's adaptive estimator in finite and continuous state spaces. Experiments in Gridworld, LunarLander, and autonomous driving using real-world logs from the Waymo Open Motion Dataset show improved estimation accuracy over alternative data collection and probe designs.
\end{abstract}

\noindent%

\medskip
\spacingset{1}

\section{Introduction}

As autonomous agents powered by large language models and reinforcement
learning techniques are being deployed at scale,
regulators, developers, and the public frequently encounter the same question:
how do we know whether these AI agents actually behave as intended?
Answering this question is critical in high-stakes domains such as autonomous driving, robotic surgery, and financial decision-making, where a gap between a system's intended and actual behavior could lead to consequential damages~\citep{kiran2021deep,liu2024survey}.
This has motivated the urgent demand for \emph{AI auditing}, which we define here as quantifying how well an AI agent would perform under its intended deployment conditions.
In light of this urgency, government regulations such as the EU AI Act require conformity assessments for high-risk AI systems before they are placed on the market \citep{euaiact2024},
and the NIST AI Risk Management Framework requires ongoing performance
monitoring of AI systems after deployment \citep{nist2023ai}.

From a technical perspective, the problem of AI auditing is to quantitatively evaluate an AI system's outputs in a variety of environments.
However, the machinery for AI auditing remains limited and largely underdeveloped \citep{raji2020closing, euaiact2024, nist2023ai}.
A standard approach to evaluating an AI system is to first deploy it in practice, and then observe its decision outcomes under a variety of practical environments. These outcomes are then aggregated into a measure of system performance.
This approach faces two practical challenges.
First, it is often prohibitively expensive to deploy a new, unverified system in real applications and collect a sufficient amount of data needed for the auditing. For example, evaluating an updated autonomous driving system across diverse traffic and terrain conditions can require extensive and costly on-road testing \citep{kiran2021deep}.
   The other challenge is how to sample auditing environments when the system's operating conditions are not arbitrary isolated states, but states generated along trajectories constrained by physical dynamics. Therefore, an auditing procedure must account not only for which environments are tested, but also for whether those environments are reachable and representative under the system's actual deployment dynamics.

A large body of work addresses these challenges through offline
policy evaluation: evaluating the quality of a target policy based on the historical observations (also known as operational logs)
from a similar AI system, such as its older version or human operation records \citep{precup2000eligibility,dudik2014doubly,uehara2020minimax,
yin2021near}.    The policy that generates these logs is called the behavior policy.    Statistical inference for policy values in infinite-horizon settings has also been developed~\citep{shi2022value}.
However, relying solely on operational logs has a fundamental limitation: when the gap in data distributions between the behavior policy and the
target policy is large, the evaluation for auditing could be arbitrarily uncertain. As a result,
seemingly comprehensive operational logs may not be adequate for accurate auditing. We refer to regions frequently visited by the target policy but poorly covered by the behavior policy as \emph{audit blind spots}. See Figure~\ref{fig:intro_concept} for an illustration.

\begin{figure}[!t]
\centering
\includegraphics[width=\textwidth]{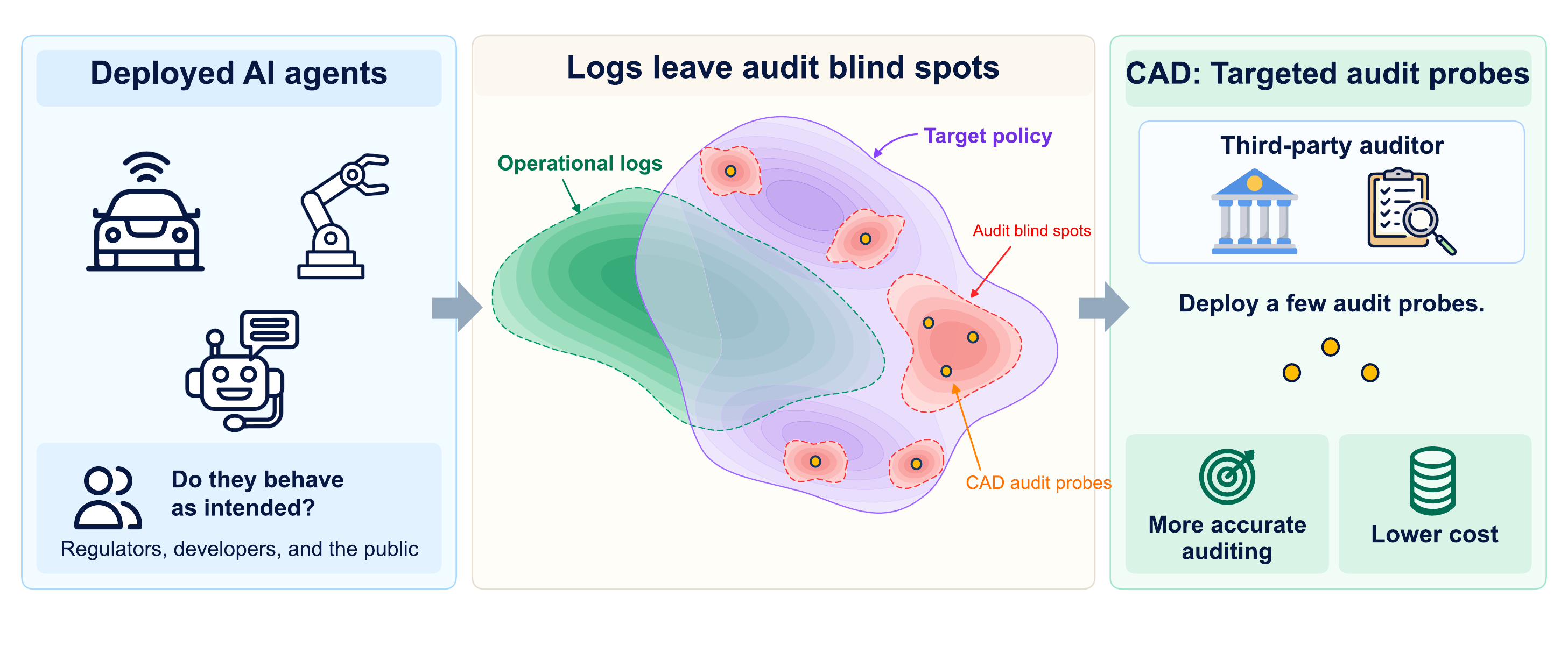}
\caption{Operational logs leave audit blind spots in the target policy's needs. CAD probes them selectively, certifying the agent at lower uncertainty and cost than logs-only auditing or direct deployment.}
\label{fig:intro_concept}
\end{figure}

To address this challenge, we propose \textbf{CAD} (Flow-\textbf{C}onstrained \textbf{A}udit \textbf{D}esign), a framework that augments fixed operational logs with targeted audit trials (probes). CAD introduces an objective that measures the mismatch between the target policy's state--action visitation distribution and that of the combined logs and probes. Minimizing this objective through a convex optimization problem subject to the environment's flow constraints yields a probe policy tailored to both the existing data coverage and the target policy.  
This allows the auditor to estimate the target policy's performance by combining existing operational logs with a limited number of targeted probes.   

Related work evaluates long-term treatment effects in sequential A/B tests using reinforcement learning~\citep{shi2023dynamic} and designs treatment allocations for precise policy comparisons in partially observable experiments~\citep{sun2026armadesign}. Here, we design additional probe trajectories to complement fixed operational logs when evaluating a specified target policy.

The proposed method has several advantages. First, from a practical perspective, CAD can improve policy-value estimation with a limited audit budget by directing probes toward state--action pairs frequently visited by the target policy but poorly covered by the operational logs. If the behavior policy never visits some of these pairs, collecting more trajectories under that same policy cannot fill the coverage gaps.  
Second, from a theoretical perspective, we relate the design objective to bounds on policy-value estimation error and extend the analysis to CAD's adaptive data collection.   Furthermore, we extend CAD to continuous state spaces using soft state aggregation and quantify the additional estimation error introduced by this aggregation.  
Third, from an empirical perspective, we evaluate CAD through simulation studies in Gridworld and LunarLander and a case study of autonomous driving using real-world logs from the Waymo Open Motion Dataset. Across these settings, CAD improves policy-value estimation relative to offline-only baselines and alternative probe designs.

The remainder of the paper is organized as follows. Section~\ref{sec:background} introduces the background on policy evaluation. Section~\ref{sec:problem} formulates the probe-design problem, and Section~\ref{sec:algorithm} presents CAD for finite and continuous state spaces. Section~\ref{sec:why_psi} establishes the theoretical guarantees. Sections~\ref{sec:simulation} and~\ref{sec:driving} report the simulation studies and the autonomous-driving case study, respectively. Section~\ref{sec:discussion} concludes the paper. Additional results and proofs are provided in the appendices.  





\section{Background}
\label{sec:background}

This section introduces Markov decision processes, policy values, and occupancy measures, then describes marginal importance sampling for policy-value estimation.   
\subsection{Markov Decision Processes and Policy Evaluation}
   We consider a \emph{discounted Markov decision process} (MDP) $\mathcal{M} = (\mathcal{S},\,\mathcal{A},\,P,\,r,\,\gamma,\,d_0)$,
where $\mathcal{S}$ is the set of possible states describing the agent's current situation, and $\mathcal{A}$ is the set of actions it can take. The \emph{transition kernel} $P(\cdot\mid s,a)$ is the probability distribution of the next state after action $a$ is taken in state $s$. The Markov assumption means that, given the current state and action, this distribution does not depend on the earlier history.
A \emph{reward} is a numerical score assigned to the outcome of an action and a state, with larger scores indicating better outcomes. The mean reward function $r:\mathcal{S}\times\mathcal{A}\to\mathbb{R}$ gives the expected immediate reward for taking action $a$ in state $s$.
The \emph{discount factor} $\gamma\in(0,1)$ determines the weight assigned to future rewards, where a reward received at step $t$ is multiplied by $\gamma^t$. The \emph{initial-state distribution} $d_0$ describes the starting conditions, with $d_0(s)$ giving the probability that the agent starts in state $s$ when interacting with the environment.
We consider both finite and continuous state spaces, with a finite action space. We first present the finite-state formulation and introduce the continuous-state extension later.          
A policy $\pi$ maps each state $s \in \mathcal{S}$ to a probability distribution $\pi(\cdot \mid s)$ over the action space $\mathcal{A}$, where $\pi(a \mid s)$ denotes the probability of selecting action $a$ when in state $s$.
Under policy $\pi$, the agent generates a trajectory
$(s_0,\,a_0,\,R_0,\;s_1,\,a_1,\,R_1,\;\ldots)
$
by drawing the initial state $s_0$ from $d_0$, then at each step $t\ge 0$ selecting action $a_t$ according to $\pi(\cdot\mid s_t)$, receiving reward $R_t$, and drawing the next state $s_{t+1}$ from $P(\cdot\mid s_t,a_t)$.
Thus, $\mathbb{E}[R_t\mid s_t=s,\,a_t=a]=r(s,a)$.
We assume $|R_t| \le R_{\max}$ almost surely, so that $\|r\|_\infty \le R_{\max}$.

Let $\mathcal{D}_{\mathrm{off}}$ denote the operational logs, consisting of $n_{\mathrm{off}}$ trajectories collected under the behavior policy $\pi_b$.    The auditor receives this dataset before the audit and cannot request additional operational logs. We let $\pi_e$ denote the target policy to be evaluated and assume that its action probabilities $\pi_e(\cdot\mid s)$ are known to the auditor.
For trajectory $i$, let $s_t^{(i)}$, $a_t^{(i)}$, and $r_t^{(i)}$ denote the observed state, action, and reward at step $t$, respectively.

For any policy $\pi$, let $(S_0, A_0, S_1, A_1, \ldots)$ denote the random trajectory generated by starting from $S_0$ drawn from $d_0$ and following $\pi$.
The discounted state-action \emph{occupancy} measure is the probability measure on $\mathcal{S}\times\mathcal{A}$ defined by
$
d^\pi(s,a)
=
(1-\gamma)\sum_{t=0}^{\infty}\gamma^t P(S_t=s,\,A_t=a),
$
which satisfies $\sum_{s \in \mathcal{S},\, a \in \mathcal{A}}d^\pi(s,a)=1$.
The occupancy measure encodes the relative frequency with which $\pi$ visits each state-action pair under discounting.
Using this occupancy measure, the expected cumulative discounted reward under policy $\pi$ can be written as
$
J(\pi)
=
\mathbb{E}_{(s,a)\sim d^\pi}\!\left[r(s,a)\right] / (1-\gamma).
$

In practice, we observe that each trajectory has at most $H$ steps, recording $(s_t,a_t,r_t,s_{t+1})$ for $t=0,\ldots,H-1$. If an episode ends before step $H$, we treat all subsequent steps as remaining in a terminal state $s_{\mathrm{term}}$ with zero reward. This convention lets us represent every trajectory over the same horizon without changing its observed return.
Over these $H$ steps, we define the \emph{$H$-step occupancy measure} as
$
d_H^\pi(s,a)
=
[(1-\gamma)/(1-\gamma^H)]\sum_{t=0}^{H-1}\gamma^t P(S_t=s,\,A_t=a),
$
which has total mass one. The \emph{$H$-step value} is the expected discounted reward over these steps:
$
J_H(\pi)
=
\mathbb{E}\Bigl[\sum_{t=0}^{H-1}\gamma^t R_t\Bigr]
=
[(1-\gamma^H)/(1-\gamma)]\,
\mathbb{E}_{(s,a)\sim d_H^\pi}\!\left[r(s,a)\right].
$
Our estimation target is the $H$-step value $J_H(\pi_e)$, which serves as the audit performance measure. Its estimate summarizes the target agent's expected discounted reward over the audit horizon.
\subsection{Importance Sampling}

We use importance sampling to estimate the target policy's expected return from trajectories collected under other policies, correcting for differences between the data-collection and target distributions. Here, we consider only marginal importance sampling (MIS), which reweights observed rewards using ratios of state--action visitation probabilities~\citep{liu2018breaking}. We first describe it for the infinite-horizon value $J(\pi_e)$.     
Let $q(s,a)$ denote the expected discounted frequency of state--action pair $(s,a)$ under the data-collection policies, normalized to sum to one. We reweight observed rewards to account for differences between this sampling occupancy and the target occupancy $d^{\pi_e}$.       Assume that $q(s,a)>0$ whenever $d^{\pi_e}(s,a)>0$, so that every pair the target policy can visit has positive probability under $q$. Define the importance weight where $q(s,a)>0$ as
$
w(s,a) \;:=\; d^{\pi_e}(s,a) / q(s,a).
$
Set $w=0$ where $q=0$, since the assumption implies $d^{\pi_e}=0$ there as well. Using $d^{\pi_e}(s,a)=q(s,a)w(s,a)$ in the definition of $J(\pi_e)$ gives
$
J(\pi_e)
=
\mathbb{E}_{(s,a)\sim q}\!\left[w(s,a)\,r(s,a)\right] / (1 - \gamma).
$
Given an estimate $\widehat w$ of $w$, the MIS estimator \citep{liu2018breaking} of $J(\pi_e)$ for infinite-horizon trajectories is
\begin{equation} \label{eq:MIS}
\widehat{J}(\widehat{w})
\;:=\;
\frac{1}{n}\sum_{i=1}^{n}\sum_{t=0}^{\infty}\gamma^t\,
\widehat{w}\!\left(s_t^{(i)},a_t^{(i)}\right)r_t^{(i)},
\end{equation}
where the outer sum runs over $n$ trajectories.

In the off-policy evaluation setting with operational logs alone \citep{uehara2020minimax}, $q=d^{\pi_b}$. If the target policy frequently visits state-action pairs that the behavior policy rarely visits, large importance weights can make estimation imprecise at the available sample size. If $d^{\pi_b}(s,a)=0$, that pair is never observed in trajectories generated by $\pi_b$. When $d^{\pi_e}(s,a)>0$ at the same pair, the support assumption is violated, and the contribution of that pair cannot be identified from the logs without additional assumptions.    Consequently, increasing the number of trajectories from $\pi_b$ alone does not guarantee that the MIS estimator converges to $J(\pi_e)$.

When trajectories are observed for $H$ steps, we instead estimate $J_H(\pi_e)$. Let $q_H$ be the $H$-step sampling occupancy and define $w_H=d_H^{\pi_e}/q_H$. With $\widehat w$ estimating $w_H$, the MIS estimator becomes
\begin{equation}\label{eq:MIS_H}
\widehat J(\widehat w)
:=\frac{1}{n}\sum_{i=1}^{n}\sum_{t=0}^{H-1}\gamma^t\,
\widehat w\!\left(s_t^{(i)},a_t^{(i)}\right)r_t^{(i)}.
\end{equation}

\section{Problem Formulation and Design Principle}
\label{sec:problem}

To address these limitations of auditing from operational logs alone, we allow the auditor to augment the fixed dataset $\mathcal{D}_{\mathrm{off}}$ with a limited number of probe trajectories collected under a probe policy $\mu$. Our task is to choose $\mu$ to improve coverage of audit blind spots, namely pairs frequently visited by $\pi_e$ but rarely observed in the logs.      In this section, we present the design under oracle conditions, assuming that the transition kernel $P$ and the occupancies $d^{\pi_e}$ and $d^{\pi_b}$ are known. For a given audit budget $B$, we choose a feasible probe occupancy to improve how well the combined logs and probes represent the target policy's state--action visitation distribution. We recover a probe policy from the selected occupancy, execute it to collect $B$ trajectories, and combine these trajectories with the operational logs to estimate the target policy's value using MIS. Figure~\ref{fig:audit_flow} summarizes this procedure.

\subsection{Effective Sampling Occupancy of Logs and Probes}
\label{sec:data}

We first consider the infinite-horizon setting without a terminal state.   
For a fixed probe policy $\mu$, let $\mathcal{D}_{\mathrm{probe}}$ denote the probe data that contain $B$ trajectories collected under $\mu$. The trajectories collected under $\mu$ use the same initial-state distribution $d_0$, transition kernel $P$, and reward mechanism as the operational logs. Together with the $n_{\mathrm{off}}>0$ trajectories from the operational logs, the combined dataset contains $n=n_{\mathrm{off}}+B$ trajectories.

We define the proportion of probe data in the combined dataset
$
\alpha = B / (n_{\mathrm{off}}+B) \in [0,1).
$
The \emph{effective sampling occupancy} of the combined dataset is the mixture
$$
q_\alpha^\mu(s,a) := (1-\alpha) d^{\pi_b}(s,a) + \alpha d^{\mu}(s,a).
$$
  When $q_\alpha^\mu$ covers the target support, the MIS estimator~\eqref{eq:MIS} uses all $n$ trajectories with ratio $w=d^{\pi_e}/q_\alpha^\mu$.

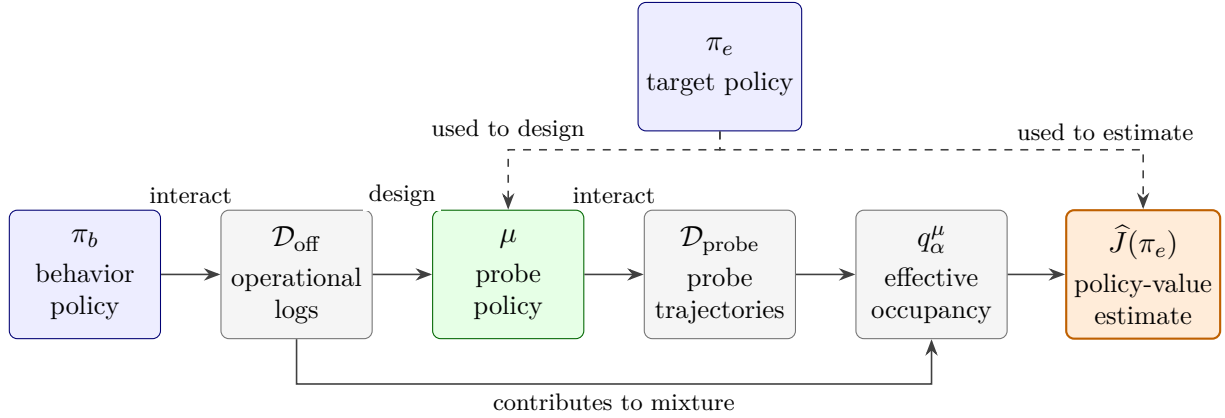
\begin{figure}[htbp]
\centering
\begingroup
\renewcommand{\baselinestretch}{1}\selectfont
\begin{tikzpicture}[
  font=\fontsize{10}{12}\selectfont,
  every node/.style={align=center},
  box/.style={rounded corners=3pt, minimum width=20mm, minimum height=17mm, inner sep=3pt},
  policy/.style={box, draw=blue!45!black, fill=blue!7},
  data/.style={box, draw=black!55, fill=gray!8},
  probe/.style={box, draw=green!40!black, fill=green!10},
  result/.style={box, draw=orange!75!black, thick, fill=orange!15},
  arr/.style={-{Stealth[length=2.2mm]}, semithick, draw=black!75},
  dashedarr/.style={arr, dashed},
  annotation/.style={font=\fontsize{9}{11}\selectfont, fill=white, inner sep=2pt}
]
\node[policy] (pe) at (8.4, 2.75) {{\normalsize $\pi_e$}\\[3pt]target policy};
\node[policy] (pb) at (0, 0) {{\normalsize $\pi_b$}\\[3pt]behavior\\policy};
\node[data] (doff) at (2.8, 0) {{\normalsize $\mathcal{D}_{\mathrm{off}}$}\\[3pt]operational\\logs};
\node[probe] (mu) at (5.6, 0) {{\normalsize $\mu$}\\[3pt]probe\\policy};
\node[data] (dprobe) at (8.4, 0) {{\normalsize $\mathcal{D}_{\mathrm{probe}}$}\\[3pt]probe\\trajectories};
\node[data] (mix) at (11.2, 0) {{\normalsize $q_\alpha^\mu$}\\[3pt]effective\\occupancy};
\node[result] (jhat) at (14, 0) {{\normalsize $\widehat{J}(\pi_e)$}\\[3pt]policy-value\\estimate};

\draw[arr] (pb) -- (doff);
\draw[arr] (doff) -- (mu);
\draw[arr] (mu) -- (dprobe);
\draw[arr] (dprobe) -- (mix);
\draw[arr] (mix) -- (jhat);
\node[annotation] at (1.4, 1.05) {interact};
\node[annotation] at (4.2, 1.05) {design};
\node[annotation] at (7.0, 1.05) {interact};

\draw[dashedarr] (pe.south) -- (8.4, 1.65) -- (5.6, 1.65) -- (mu.north);
\draw[dashedarr] (pe.south) -- (8.4, 1.65) -- (14, 1.65) -- (jhat.north);
\node[annotation, above] at (5.6, 1.7) {used to design};
\node[annotation, above] at (13.5, 1.7) {used to estimate};

\draw[arr] (doff.south) -- (2.8, -1.45) -- (11.2, -1.45) -- (mix.south);
\node[annotation, below] at (7, -1.5) {contributes to mixture};
\end{tikzpicture}
\endgroup
\caption{Audit problem data flow. The auditor receives fixed operational logs $\mathcal{D}_{\mathrm{off}}$ previously collected under behavior policy $\pi_b$.    For a target policy $\pi_e$, the auditor uses these logs to design a probe policy $\mu$ and executes it in the environment to collect probes $\mathcal{D}_{\mathrm{probe}}$.    The mixture of $\mathcal{D}_{\mathrm{off}}$ and $\mathcal{D}_{\mathrm{probe}}$ has effective sampling occupancy $q_\alpha^\mu$, and combined with $\pi_e$ yields the policy-value estimate $\widehat{J}(\pi_e)$.    }
\label{fig:audit_flow}
\end{figure}

\subsection{Occupancy Mismatch as the Probe Design Criterion}

To design probes that complement the fixed operational logs, we propose the following criterion for measuring the mismatch between the combined sampling occupancy $q$ and the target occupancy:
\[
\Psi(q)
\;=\;
\sum_{s \in \mathcal{S},\, a \in \mathcal{A}}\frac{d^{\pi_e}(s,a)^2}{q(s,a)}
\;=\;
\sum_{s \in \mathcal{S},\, a \in \mathcal{A}}\frac{(d^{\pi_e}(s,a)-q(s,a))^2}{q(s,a)} + 1
\;=\;
\chi^2(d^{\pi_e}\,\|\,q) + 1,
\]

where $\chi^2$ denotes the chi-squared divergence. Here $\Psi(q)$ measures the mismatch between the sampling occupancy $q$ and the target $d^{\pi_e}$. The criterion is minimized at $\Psi(q)=1$ when the sampling occupancy exactly matches the target occupancy. It strongly penalizes insufficient coverage of target-relevant state--action pairs: if $d^{\pi_e}(s,a)>0$, the corresponding term grows without bound as $q(s,a)$ approaches zero.
In Section~\ref{sec:var_just}, we derive an upper bound on the MIS estimator's mean squared error (MSE) that is proportional to $\Psi(q)/n$.

\subsection{Flow-Constrained Probe Design}

Probe design must also respect the occupancies attainable under the environment dynamics.    A probe cannot allocate occupancy arbitrarily: its trajectories must follow the transition kernel $P$ from $d_0$. We express this restriction through the \emph{Bellman flow-feasible set} \citep{puterman2014markov}:
\begin{equation}\label{eq:flow_constraint}
\mathcal{D}_{\mathrm{flow}}
=
\bigl\{
x\ge 0 :
\sum_{a \in \mathcal{A}}\!x(s,a)
=
(1-\gamma)d_0(s)
+\gamma\!\sum_{s' \in \mathcal{S},a' \in \mathcal{A}}\!x(s',a')P(s\mid s',a'),\;
\forall s\in\mathcal{S}
\bigr\},
\end{equation}
where $x : \mathcal{S}\times\mathcal{A}\to\mathbb{R}_{\ge 0}$ is a candidate occupancy measure over state-action pairs.
\begin{remark}[Flow-feasibility characterization]
\label{prop:flow-feasibility}
The flow constraints express a balance at each state: its total discounted occupancy equals the contribution from the initial-state distribution plus the discounted arrivals from other state--action pairs. They therefore connect the allocation of probe visits to the environment dynamics. More precisely, a vector $x \ge 0$ over $\mathcal{S}\times\mathcal{A}$ is the occupancy measure of some stationary Markov policy if and only if $x \in \mathcal{D}_{\mathrm{flow}}$ \citep[Chapter~6]{puterman2014markov}. This characterization provides both a restriction and a guarantee for probe design. Every stationary probe policy must induce an occupancy in $\mathcal{D}_{\mathrm{flow}}$, and every occupancy in this set can be realized by such a policy. Since the design criterion depends on the probe policy through its occupancy, we can first optimize the visitation pattern over this feasible set and then recover a policy that realizes it.
\end{remark}

We thus formulate the probe design problem as the following optimization:
\begin{equation}\label{eq:probe_design}
x^* \in \argmin_{x\,\in\,\mathcal{D}_{\mathrm{flow}}}
\Psi\!\left((1-\alpha)d^{\pi_b}+\alpha x\right).
\end{equation}
Here, $x^*$ is an optimal probe occupancy. The design problem is convex and admits a global minimizer for $B>0$.
Once the optimal occupancy $x^*$ is obtained, the corresponding probe policy is recovered analytically via
$$
\mu^*(a\mid s) = x^*(s,a) / \sum_{a' \in \mathcal{A}} x^*(s,a'),
$$
which is well-defined whenever $\sum_{a \in \mathcal A} x^*(s,a)>0$.
The auditor executes $\mu^*$ in the environment to collect probe trajectories. The auditor then forms the combined dataset of size $n = n_{\mathrm{off}} + B$, and applies the MIS estimator \eqref{eq:MIS}
with effective sampling occupancy $q_\alpha^{\mu^*}$ to obtain an estimate of $J(\pi_e)$.
The minimum mismatch achievable by this design is limited by the fixed contribution of the operational logs. For a given audit budget, the combined occupancy satisfies $q(s,a)\ge(1-\alpha)d^{\pi_b}(s,a)$, so exact matching requires the target occupancy to be at least this large at every state--action pair. Appendix~\ref{app:design_limits} shows that this condition is also sufficient and provides bounds on the smallest achievable mismatch.   Appendices~\ref{app:kkt} and~\ref{app:probing_gain} further characterize the optimal probe design and the reduction in mismatch achievable with a small audit budget, respectively.    

\section{Practical Algorithm}
\label{sec:algorithm}

We now implement the oracle design of Section~\ref{sec:problem} when $P$, $d^{\pi_e}$, and $d^{\pi_b}$ are unknown.
   In practice, probe design depends on estimates of these quantities, while each new probe provides observations that can improve those estimates. To update the design as these estimates change, we propose CAD, an iterative procedure that repeatedly estimates the unknown quantities, solves the empirical design problem, collects new probes, and refines the design.

\begingroup
The first two subsections consider finite state spaces, and the third extends CAD to continuous state spaces.
Each trajectory is represented by $H$ transitions $(s_t,a_t,r_t,s_{t+1})$, for $t=0,\ldots,H-1$. 
The algorithm proceeds in $K$ rounds, where $1\le K\le B$. We allocate $B_k=\lfloor B/K\rfloor$ trajectories for $k=0,\ldots,K-2$ and $B_{K-1}=B-(K-1)\lfloor B/K\rfloor$ to the final round, so that $\sum_kB_k=B$.
We initialize $\mathcal D_0=\mathcal D_{\mathrm{off}}$ and estimate $d_H^{\pi_b}$ once from the operational logs, before collecting any probes, by normalizing their discounted visitation counts:
\begin{equation}
\label{eq:empirical_behavior_occupancy}
\widehat d^{\pi_b}(s,a)
=\frac{1-\gamma}{n_{\mathrm{off}}(1-\gamma^H)}
\sum_{i=1}^{n_{\mathrm{off}}}\sum_{t=0}^{H-1}
\gamma^t\,\mathbf{1}[s_t^{(i)}=s,\,a_t^{(i)}=a].
\end{equation}
We set $\hat A_0=n_{\mathrm{off}}\widehat d^{\pi_b}$ as the initial accumulated occupancy.
\endgroup

\subsection{Iterative Probe Design and Data Collection}
\label{sec:opt_perspective}

\begingroup
At the start of round $k=0,\ldots,K-1$, $\mathcal D_k$ contains the operational logs and all previously collected probe trajectories.
To estimate $P$, let $N_k(s,a,s')$ count transitions from state $s$ to state $s'$ under action $a$ in $\mathcal D_k$. For state--action pairs with positive total counts, the empirical transition estimate is
\begin{equation}
\label{eq:empirical_transition}
\widehat P_k(s'\mid s,a)
=\frac{N_k(s,a,s')}{\sum_{u\in\mathcal S}N_k(s,a,u)}.
\end{equation}
For a state--action pair with no observed transition, we set $\widehat P_k(s'\mid s,a)=1/|\mathcal S|$ for every $s'\in\mathcal S$ as a default estimate in the absence of prior information.  If a designated terminal state is present, its rows are instead fixed at $\widehat P_k(s_{\mathrm{term}}\mid s_{\mathrm{term}},a)=1$ for every action.
Given $\widehat P_k$ and the known target policy $\pi_e$, we compute the $H$-step target occupancy $\widehat d_{H,k}^{\pi_e}$ by propagating the state distribution forward under $\widehat P_k$ \citep{puterman2014markov}. Specifically, set $\hat p_{0,k}=d_0$ and, for $t=0,\ldots,H-2$,
$
\hat p_{t+1,k}(s')=\sum_{s\in\mathcal S,\,a\in\mathcal A}\hat p_{t,k}(s)\,\pi_e(a\mid s)\,\widehat P_k(s'\mid s,a).
$
Then
\begin{equation}
\label{eq:target_occupancy_H}
\widehat d_{H,k}^{\pi_e}(s,a)
=\frac{1-\gamma}{1-\gamma^H}\sum_{t=0}^{H-1}\gamma^t\,\hat p_{t,k}(s)\,\pi_e(a\mid s),
\qquad (s,a)\in\mathcal S\times\mathcal A.
\end{equation}
\endgroup

Using these estimates, we choose a probe to supplement the data already collected.
If a designated terminal state is present, we use the zero-reward terminal state $s_{\mathrm{term}}$ introduced in Section~\ref{sec:background} and write $\mathcal S^{-}=\mathcal S\setminus\{s_{\mathrm{term}}\}$; otherwise, $\mathcal S^{-}=\mathcal S$.
The empirical design objective $\Psi$ sums over $\mathcal S^{-}$ because the terminal state contributes no reward, while the flow constraints retain all states.
Let $\hat A_k$ denote the accumulated empirical occupancy of the data collected before round $k$.
Using $\hat A_k$, we adapt the probe-design problem~\eqref{eq:probe_design} to round $k$ as
\begingroup
\begin{equation}
\label{eq:emp_obj_iter}
\hat{x}_k^*
\;\in\;
\argmin_{x\,\in\,\hat{\mathcal{D}}_{\mathrm{flow},k}}
\sum_{s \in \mathcal{S}^{-},\, a \in \mathcal{A}}
\frac{\widehat{d}_{H,k}^{\pi_e}(s,a)^2}{\hat{A}_k(s,a) + B_k\,x(s,a) + \varepsilon},
\end{equation}
where $\hat{\mathcal{D}}_{\mathrm{flow},k}$ is obtained from~\eqref{eq:flow_constraint} by replacing $P$ with $\widehat P_k$, and $\varepsilon>0$ is a small regularizer that keeps the denominator positive.
\endgroup
The objective is convex in $x$ and the Bellman flow constraints are linear, so problem~\eqref{eq:emp_obj_iter} is a convex program and admits a global minimizer, which can be computed with a standard convex solver.
The corresponding probe policy on $\mathcal S^{-}$ is recovered by row-sum normalization, with a uniform action distribution when the design assigns zero mass to a state. At $s_{\mathrm{term}}$, actions have no effect, and we use the same arbitrary fixed action for all policies. Specifically, the probe policy is given by
\begingroup
\begin{equation}\label{eq:probe_policy_recovery}
\mu_k^*(a\mid s)=
\begin{cases}
\displaystyle\frac{\hat x_k^*(s,a)}{\sum_{a'\in\mathcal A}\hat x_k^*(s,a')},
&\displaystyle\sum_{a'\in\mathcal A}\hat x_k^*(s,a')>0,\\[6pt]
1/|\mathcal A|, &\text{otherwise},
\end{cases}
\qquad s\in\mathcal S^{-}.
\end{equation}
\endgroup

\begingroup
The auditor then executes $\mu_k^*$ for $B_k$ trajectories to form $\mathcal D_{\mathrm{on}}^{(k)}$ and updates $\mathcal D_{k+1}=\mathcal D_k\cup\mathcal D_{\mathrm{on}}^{(k)}$. After collecting these trajectories, we update the accumulated occupancy as
\begin{equation}
\label{eq:accumulated_occupancy_update}
\hat A_{k+1}(s,a)=\hat A_k(s,a)+B_k\,\widehat d^{\mu_k^*}(s,a),
\end{equation}
where $\widehat d^{\mu_k^*}$ is the normalized discounted empirical occupancy estimated from trajectories collected under probe policy $\mu_k^*$, computed as in~\eqref{eq:empirical_behavior_occupancy}.
Through these updates, the accumulated denominator $\hat{A}_{k+1}$ encodes how much evidence has already been gathered at each state-action pair and helps direct new probes toward the remaining audit blind spots.
\endgroup

\subsection{Final Policy-Value Estimation}
\label{sec:procedure}

\begingroup
After all $K$ rounds, we estimate $\widehat P_K$ from all data $\mathcal D_K$ and compute $\widehat d_{H,K}^{\pi_e}$ by~\eqref{eq:target_occupancy_H}. The estimated final sampling occupancy is $\widehat q_K(s,a) = \hat{A}_K(s,a)/(n_{\mathrm{off}} + B)$.
\endgroup
We apply the $H$-step MIS estimator~\eqref{eq:MIS_H} to the combined data with regularized marginal importance weights:
\begin{equation}
\label{eq:mis_final}
\widehat{w}_K (s,a) \;=\; \frac{\widehat{d}^{\pi_e}_{H,K}(s,a)}{\widehat q_K(s,a) + \lambda},
\qquad
\widehat{J}(\widehat{w}_K)
\;=\; \frac{1}{n}
\sum_{i=1}^{n}\sum_{t=0}^{H-1}
\gamma^t\,\widehat{w}_K (s_{t}^{(i)},a_{t}^{(i)})\,r_{t}^{(i)},
\end{equation}
\begingroup
where $\lambda > 0$ is a small regularizer, $n = n_{\mathrm{off}} + B$, and the outer sum runs over all collected trajectories, both offline and online.
\endgroup
Algorithm~\ref{alg:probe_design} summarizes the procedure. 

\begin{algorithm}[!htbp]
\small
\caption{CAD on a finite state space}
   \label{alg:probe_design}
\begin{algorithmic}[1]
\Require Operational logs $\mathcal{D}_{\mathrm{off}}$ ($n_{\mathrm{off}}$ trajectories), target policy $\pi_e$, initial distribution $d_0$, discount $\gamma$, horizon $H$, total audit budget $B$, rounds $1\le K\le B$, regularizers $\varepsilon, \lambda>0$.
\Ensure Probe policies $\{\mu_k^*\}_{k=0}^{K-1}$ and policy-value estimate $\widehat{J}$.

\State \textbf{Initialize.} Set $\mathcal{D}_0 \leftarrow \mathcal{D}_{\mathrm{off}}$; compute the discounted empirical occupancy $\widehat{d}^{\pi_b}$ from $\mathcal{D}_{\mathrm{off}}$; set $\hat{A}_0 \leftarrow n_{\mathrm{off}}\,\widehat{d}^{\pi_b}$; set $B_k=\lfloor B/K\rfloor$ for $k<K-1$ and $B_{K-1}=B-(K-1)\lfloor B/K\rfloor$.

\For{$k = 0, 1, \ldots, K-1$}
    \State \textbf{Estimate dynamics.} Compute $\widehat{P}_k$ from the trajectory counts in $\mathcal{D}_k$ using~\eqref{eq:empirical_transition}.
    \State \textbf{Estimate target occupancy.} Compute the $H$-step occupancy $\widehat{d}^{\pi_e}_{H,k}$ induced by $\pi_e$ under $\widehat{P}_k$ via the forward recursion~\eqref{eq:target_occupancy_H}.
    \State \textbf{Probe design.} Solve~\eqref{eq:emp_obj_iter} for $\hat{x}^*_k$; recover $\mu_k^*$ by~\eqref{eq:probe_policy_recovery}.
       \State \textbf{Collect probes.} Execute $\mu_k^*$ for $B_k$ trajectories to form $\mathcal{D}^{(k)}_{\mathrm{on}}$; update $\mathcal{D}_{k+1} \leftarrow \mathcal{D}_k \cup \mathcal{D}^{(k)}_{\mathrm{on}}$.
    \State \textbf{Update denominator.} Estimate the discounted empirical occupancy $\widehat{d}^{\mu_k^*}$ from $\mathcal{D}^{(k)}_{\mathrm{on}}$; set $\hat{A}_{k+1} \leftarrow \hat{A}_k + B_k\,\widehat{d}^{\mu_k^*}$.
   \EndFor

\State \textbf{Final estimation.} Estimate $\widehat{P}_K$ and $\widehat{d}^{\pi_e}_{H,K}$ from all data $\mathcal{D}_K = \mathcal{D}_{\mathrm{off}} \cup \bigcup_{k=0}^{K-1}\mathcal{D}^{(k)}_{\mathrm{on}}$.
\State Compute $\widehat q_K = \hat{A}_K/(n_{\mathrm{off}} + B)$ and $\widehat{w}_K$ via \eqref{eq:mis_final}.
\State Return $\widehat{J}(\widehat{w}_K)$ via \eqref{eq:mis_final} over all $n = n_{\mathrm{off}} + B$ trajectories.

\end{algorithmic}
\end{algorithm}

\subsection{Extension of CAD to Continuous State Spaces}
\label{sec:continuous_algorithm}

For a continuous state space, we use soft state aggregation \citep{bezdek1981pattern,li2006towards} to assign each state to a finite set of cells and apply CAD to these cells. We then construct a probe policy on the original state space by combining the cell-level action probabilities. In the finite-state setting, we compute the target occupancy by propagating state probabilities through the estimated transition matrix. For continuous states, we approximate this computation using trajectories simulated under the target policy $\pi_e$ in fitted transition models.

We fit $M$ transition models to the available environment observations $\mathcal D$ and generate $N_{\mathrm{sim}}$ trajectories from each, averaging across model fits and simulated trajectories to stabilize the estimate. Each model predicts the next state from the current state and action.      The simulated trajectories start from states sampled from the offline initial states and stop at termination or step $H$. We pool the resulting $M N_{\mathrm{sim}}$ trajectories into a dataset denoted by $\textsc{Sim}(\pi_e,\mathcal D)$.

We construct $K_\varphi$ cells for non-terminal states from $\mathcal D_{\mathrm{off}}\cup\mathcal S_{\mathrm{init}}$, where $\mathcal S_{\mathrm{init}}=\textsc{Sim}(\pi_e,\mathcal D_{\mathrm{off}})$, and add a separate cell $0$ for the terminal state. Let $\mathcal J=\{0,1,\ldots,K_\varphi\}$ index all cells. For a non-terminal state $s$, let $\varphi_j(s)\ge0$ denote its weight in cell $j$, with $\sum_{j=1}^{K_\varphi}\varphi_j(s)=1$. The weights remain fixed across rounds. Appendix~\ref{app:soft_state_aggregation} gives the clustering and weight construction.

We use the superscript $\varphi$ to identify occupancies on the cells. For a dataset $\mathcal D$ of $|\mathcal D|$ trajectories, we compute empirical occupancies as in~\eqref{eq:empirical_behavior_occupancy}, replacing state indicators by cell weights:
\begin{equation}\label{eq:continuous_emp_occ}
\widehat d_{H,\mathcal D}^{\varphi}(j,a)
=\frac{1-\gamma}{|\mathcal D|(1-\gamma^H)}\sum_{i\in\mathcal D}\sum_{t=0}^{H-1}
\gamma^t\varphi_j(s_t^{(i)})\mathbf 1\{a_t^{(i)}=a\}.
\end{equation}
At round $k$, we generate $\mathcal S_k=\textsc{Sim}(\pi_e,\mathcal D_k)$ and estimate the target occupancy by $\widehat d_{H,k}^{\pi_e,\varphi}=\widehat d_{H,\mathcal S_k}^{\varphi}$.

To estimate $\widehat P_k$, we count transitions between cells in the real and simulated trajectories $\mathcal D_k\cup\mathcal S_k$, weighting each transition by the cell weights of its current and next states. We estimate $\widehat d_{0,k}$ by averaging the cell weights of the initial states in the real trajectories $\mathcal D_k$. Appendix~\ref{app:cell_transitions} gives the details.

Using these estimates, we solve the cell-level design problem
\begin{equation}\label{eq:continuous_design}
\begin{aligned}
\hat x_k\in\argmin_{x\ge0}\quad
&\sum_{j=1}^{K_\varphi}\sum_{a\in\mathcal A}
\frac{\widehat d_{H,k}^{\pi_e,\varphi}(j,a)^2}{\hat A_k(j,a)+(B_k/(1-\gamma^H))x(j,a)+\varepsilon}\\
\text{subject to}\quad
&\sum_{a\in\mathcal A}x(j,a)=(1-\gamma)\widehat d_{0,k}(j)
+\gamma\sum_{\ell\in\mathcal J}\sum_{a\in\mathcal A}x(\ell,a)\widehat P_k(j\mid\ell,a),\quad j\in\mathcal J,
\end{aligned}
\end{equation}
We initialize $\hat A_0=n_{\mathrm{off}}\widehat d_{H,\mathcal D_{\mathrm{off}}}^{\varphi}$ from the offline logs and update $\hat A_{k+1}=\hat A_k+B_k\widehat d_{H,\mathcal D_{\mathrm{on}}^{(k)}}^{\varphi}$ after each round using the newly collected probe trajectories, with empirical occupancies computed by~\eqref{eq:continuous_emp_occ}.

We recover the cell-level policy $\mu_k(a\mid j)$ by~\eqref{eq:probe_policy_recovery}. At a continuous state $s$, the probe selects an action according to
\begin{equation}\label{eq:continuous_probe_exec}
\tilde\mu_k(a\mid s)=\sum_{j=1}^{K_\varphi}\varphi_j(s)\mu_k(a\mid j).
\end{equation}

After each round, we add the new probe trajectories to $\mathcal D_k$ and update the sampling occupancy from these real observations. After the final round, we refit the models on $\mathcal D_K$ and estimate $\widehat d_{H,K}^{\pi_e,\varphi}$ from new simulated target trajectories. With $n=n_{\mathrm{off}}+B$ and $\widehat q_K^{\varphi}=\hat A_K/n$, we average the cell-level importance weights at each observed state and apply~\eqref{eq:mis_final}:
\begin{equation}\label{eq:continuous_mis_final}
\begin{aligned}
\widehat w_K^{\varphi}(j,a)&=\frac{\widehat d_{H,K}^{\pi_e,\varphi}(j,a)}{\widehat q_K^{\varphi}(j,a)+\lambda},
\qquad j=1,\ldots,K_\varphi,\\
\widehat J(\widehat w_K^{\varphi})&=\frac1n\sum_{i\in\mathcal D_K}\sum_{t=0}^{H-1}\gamma^t r_t^{(i)}
\sum_{j=1}^{K_\varphi}\varphi_j(s_t^{(i)})\widehat w_K^{\varphi}(j,a_t^{(i)}).
\end{aligned}
\end{equation}
Algorithm~\ref{alg:cad_continuous} summarizes the procedure.    

\begin{algorithm}[!htbp]
\small
\caption{CAD on a continuous state space}
\label{alg:cad_continuous}
\begin{algorithmic}[1]
\Require Operational logs $\mathcal D_{\mathrm{off}}$ ($n_{\mathrm{off}}$ trajectories), target policy $\pi_e$, discount $\gamma$, horizon $H$, total audit budget $B$, rounds $1\le K\le B$, regularizers $\varepsilon,\lambda>0$, number of state cells $K_\varphi$, simulation parameters $M,N_{\mathrm{sim}}$.
\Ensure Policy-value estimate $\widehat J$.
   \State \textbf{Fix abstraction.} Generate $\mathcal S_{\mathrm{init}}\gets\textsc{Sim}(\pi_e,\mathcal D_{\mathrm{off}})$; construct and fix $\varphi$ from $\mathcal D_{\mathrm{off}}\cup\mathcal S_{\mathrm{init}}$ as in Appendix~\ref{app:soft_state_aggregation}.
\State \textbf{Initialize.} Set $\mathcal D_0\gets\mathcal D_{\mathrm{off}}$; compute $\widehat d_{H,\mathcal D_{\mathrm{off}}}^{\varphi}$ using~\eqref{eq:continuous_emp_occ}; set $\hat A_0\gets n_{\mathrm{off}}\widehat d_{H,\mathcal D_{\mathrm{off}}}^{\varphi}$; set $B_k=\lfloor B/K\rfloor$ for $k<K-1$ and $B_{K-1}=B-(K-1)\lfloor B/K\rfloor$.
\For{$k=0,\ldots,K-1$}
  \State \textbf{Estimate target occupancy.} Generate $\mathcal S_k\gets\textsc{Sim}(\pi_e,\mathcal D_k)$; compute $\widehat d_{H,k}^{\pi_e,\varphi}$ from $\mathcal S_k$ using~\eqref{eq:continuous_emp_occ}.
  \State \textbf{Estimate dynamics.} Compute $\widehat P_k$ from $\mathcal D_k\cup\mathcal S_k$ and $\widehat d_{0,k}$ from the first states of $\mathcal D_k$, as in Appendix~\ref{app:cell_transitions}.
  \State \textbf{Probe design.} Solve~\eqref{eq:continuous_design} for $\hat x_k$; recover $\mu_k$ by~\eqref{eq:probe_policy_recovery}.
  \State \textbf{Collect probes.} Execute~\eqref{eq:continuous_probe_exec} for $B_k$ real trajectories to form $\mathcal D_{\mathrm{on}}^{(k)}$; set $\mathcal D_{k+1}\gets\mathcal D_k\cup\mathcal D_{\mathrm{on}}^{(k)}$.
  \State \textbf{Update denominator.} Compute $\widehat d_{H,\mathcal D_{\mathrm{on}}^{(k)}}^{\varphi}$ using~\eqref{eq:continuous_emp_occ}; set $\hat A_{k+1}\gets\hat A_k+B_k\widehat d_{H,\mathcal D_{\mathrm{on}}^{(k)}}^{\varphi}$.
\EndFor
\State \textbf{Final estimation.} Generate $\mathcal S_K\gets\textsc{Sim}(\pi_e,\mathcal D_K)$; compute $\widehat d_{H,K}^{\pi_e,\varphi}$ using~\eqref{eq:continuous_emp_occ}.
\State Compute $\widehat q_K^{\varphi}=\hat A_K/(n_{\mathrm{off}}+B)$ and $\widehat w_K^{\varphi}$ via~\eqref{eq:continuous_mis_final}.
\State Return $\widehat J(\widehat w_K^{\varphi})$ via~\eqref{eq:continuous_mis_final} over all $n=n_{\mathrm{off}}+B$ real trajectories.
\end{algorithmic}
\end{algorithm}

\subsection{Tuning Parameters and Implementation of CAD}
\label{sec:hyperparams}

For CAD on a finite state space, we choose three tuning parameters: the number of rounds $K$ and the regularizers $\lambda$ and $\varepsilon$.    The number of rounds balances adaptivity and estimation precision. Too few rounds limit opportunities to refine the design, whereas too many reduce the per-round budget $\lfloor B/K\rfloor$, making occupancy estimates less reliable for subsequent design updates. In practice, we set $K=3$ to balance adaptivity and estimation precision.
The regularizers $\lambda$ and $\varepsilon$ stabilize the denominators in the importance weights and the design objective, respectively. A larger $\lambda$ reduces sensitivity to occupancy estimation error but introduces regularization bias in $\widehat{J}$. A larger $\varepsilon$ reduces the design objective's sensitivity to estimation error but also weakens its emphasis on poorly covered state--action pairs. Theorem~\ref{thm:final-mse} quantifies how regularization and occupancy estimation error enter the final mean squared error bound.
In practice, we use a base regularization level of $10^{-6}$ for $\lambda$ and $\varepsilon$ to stabilize the importance weights and the probe-design objective.

For continuous states, $K_\varphi$ controls the resolution of the state representation. More cells distinguish finer state differences but leave fewer observations for estimating each cell's transitions, so the choice should reflect the amount of available data. The model count $M$ and the number of simulated trajectories per model $N_{\mathrm{sim}}$ balance estimation stability and computational cost. In practice, we use $M=5$ transition models and a few hundred simulated trajectories per model.


\section{Theoretical Analysis}
\label{sec:why_psi}

Section~\ref{sec:problem} introduced $\Psi(q_\alpha^\mu)$ as an occupancy-mismatch criterion for designing probe policies.
This section relates the design criterion to statistical upper bounds and analyzes the estimation error of CAD in finite and continuous state spaces.
Section~\ref{sec:var_just} relates $\Psi(q_\alpha^\mu)$ to mean squared error and asymptotic variance upper bounds under a fixed probe policy.
Section~\ref{sec:adaptive_plugin_bound} specializes the same principle to the final policy-value estimator in~\eqref{eq:mis_final}, computed by Algorithm~\ref{alg:probe_design}.
   Section~\ref{sec:continuous_extension} extends the error analysis to continuous states through the cell-level representation. Proofs of the theoretical results are collected in Appendix~\ref{app:proofs}.

\subsection{Statistical Justification of the Design Criterion}
   \label{sec:var_just}

We first consider a fixed probe policy $\mu$ in the infinite-horizon setting of Section~\ref{sec:problem}. The data comprise $n_{\mathrm{off}}$ independent trajectories from $\pi_b$ and $B$ independent trajectories from $\mu$, allowing dependence between the two samples. Write $n=n_{\mathrm{off}}+B$ and assume that $q_\alpha^\mu$ is positive on $\mathrm{supp}(d^{\pi_e})$. Observed rewards satisfy $|R_t|\le R_{\max}$ almost surely, as in Section~\ref{sec:background}. The following result bounds the mean squared error (MSE) of the plug-in MIS estimator, allowing its weights to be estimated from the evaluation data.

\begin{theorem}[MSE bound for the plug-in MIS estimator]
\label{thm:mse_decomp}
Let $\widehat w=\hat d^{\pi_e}/\hat q$ and let $\widehat J(\widehat w)$ be the estimator in~\eqref{eq:MIS}, with zero weights outside $\mathrm{supp}(d^{\pi_e})$. Suppose that, almost surely, (i) $\mathrm{supp}(\hat d^{\pi_e})\subseteq\mathrm{supp}(d^{\pi_e})$ and $\hat q\ge c\,q_\alpha^\mu$ on the target support for a constant $c\in(0,1]$, and (ii) $|\hat d^{\pi_e}-d^{\pi_e}|\le\rho_p\,d^{\pi_e}$ and $|\hat q-q_\alpha^\mu|\le\rho_q\,q_\alpha^\mu$ pointwise on $\mathrm{supp}(d^{\pi_e})$, where $\rho_p,\rho_q\ge0$ are deterministic and $\rho_p^2+\rho_q^2\le C_0/n$ for a constant $C_0>0$. Then there is a constant $C>0$, depending only on $\gamma,c,C_0$, and $R_{\max}$, such that
\[
\mathbb E\!\left[\bigl(\widehat J(\widehat w)-J(\pi_e)\bigr)^2\right]
\le\frac{C\,\Psi(q_\alpha^\mu)}{n}.
\]
\end{theorem}

For fixed $n$, minimizing $\Psi$ minimizes this upper bound when the error controls hold uniformly over the candidate designs with common constants $c$ and $C_0$.

Under additional conditions, $\widehat J(\widehat w)$ is asymptotically normal as $n\to\infty$ with $\alpha\in(0,1)$ fixed.     The variance $\sigma^2$ of the limit of $\sqrt n(1-\gamma)\{\widehat J(\widehat w)-J(\pi_e)\}$ satisfies
$
\sigma^2\le R_{\max}^2\Psi(q_\alpha^\mu).
$
Appendices~\ref{app:clt_conditions} and~\ref{app:variance_char} give the central limit theorem and the variance upper bound. These results justify $\Psi$ as a criterion for controlling error bounds under fixed designs. The next subsection treats the adaptive procedure, including reuse of all collected trajectories in the final estimate.

\begingroup
\subsection{Error Bound for Adaptive Plug-in Policy Evaluation}
\label{sec:adaptive_plugin_bound}

In CAD, each probe policy is chosen using data from earlier rounds. We analyze the $H$-step estimator in~\eqref{eq:mis_final}, whose target is $J_H(\pi_e)$. The following bound shows how estimation error depends on data coverage and errors in estimating the occupancies.

Let $H$ be the observed horizon, let $B_0,\ldots,B_{K-1}$ be fixed audit budgets, and write $n=n_{\mathrm{off}}+\sum_{k=0}^{K-1}B_k$. Define $q_{H,K}=[n_{\mathrm{off}}d_H^{\pi_b}+\sum_{k=0}^{K-1}B_kd_H^{\mu_k}]/n$.
This is the realized mixture of $H$-step occupancies and can be random because the probes depend on previously collected data. The algorithm estimates it by $\widehat q_K=\hat A_K/n$. Write $\delta_{H,K}=\|\widehat d^{\pi_e}_{H,K}-d_H^{\pi_e}\|_1$ and $\tau_{H,K}=\|\widehat q_K-q_{H,K}\|_\infty$ for the target and sampling occupancy errors. For $\lambda>0$, define the regularized design criterion and the regularization bias as
\begingroup
\[
\Psi_{H,\lambda}(q):=\sum_{s,a}\frac{d_H^{\pi_e}(s,a)^2}{q(s,a)+\lambda},
\qquad
\mathrm{Bias}_{H,\lambda}(q):=-\frac{1-\gamma^H}{1-\gamma}\sum_{s,a}
\frac{\lambda\,d_H^{\pi_e}(s,a)}{q(s,a)+\lambda}\,r(s,a).
\]
\endgroup
\begin{theorem}[Error bound for the adaptive audit estimator]
\label{thm:final-mse}
Consider a finite MDP with rewards satisfying $|R_t|\le R_{\max}$. The offline trajectories are independent and generated by $\pi_b$. At round $k$, the probe policy $\mu_k$ is selected based on data collected before that round, then $B_k$ new trajectories are collected under $\mu_k$. All trajectories are observed for $H$ steps.

Suppose that, with probability one, $\widehat d^{\pi_e}_{H,K}(s,a)=0$ whenever $d_H^{\pi_e}(s,a)=0$. Consider the estimator in~\eqref{eq:mis_final} with $\lambda>0$. Write $m=|\mathrm{supp}(d_H^{\pi_e})|$. For any $\zeta\in(0,1)$, set $\ell=\log\{2\mathrm e\,m[1+\lceil\log_2 n\rceil]/\zeta\}$. Then
\[
\begin{aligned}
&\mathbb E\!\left[\bigl(\widehat J(\widehat w_K)-J_H(\pi_e)\bigr)^2\right]
\le R_{\max}^2\left(\frac{1-\gamma^H}{1-\gamma}\right)^2
\Biggl\{\frac{24m\ell}{n}\,\mathbb E[\Psi_{H,\lambda}(q_{H,K})]\\
&\qquad+\frac{96\ell^2}{n^2\lambda^2}
+\zeta(1+1/\lambda)^2+3\mathbb E\!\left[
\left(\frac{\delta_{H,K}}{\lambda}+\frac{\tau_{H,K}}{\lambda^2}\right)^2\right]
\Biggr\}
+3\mathbb E\bigl[\mathrm{Bias}_{H,\lambda}(q_{H,K})^2\bigr].
\end{aligned}
\]
All expectations are over the collected data and algorithmic randomization.
\end{theorem}

   The bound depends on the achieved sampling occupancy through $\Psi_{H,\lambda}(q_{H,K})$, and on errors in estimating the target and sampling occupancies through $\delta_{H,K}$ and $\tau_{H,K}$. For a fixed $\zeta$, the factor $\ell$ grows slowly with the number of target state--action pairs $m$ and the sample size $n$. We can choose $\zeta\in(0,1)$ when evaluating the bound: a smaller $\zeta$ reduces the term $\zeta(1+1/\lambda)^2$ but increases the terms involving $\ell$. The following corollary gives sufficient conditions for $\widehat J(\widehat w_K)$ to converge in mean square to $J_H(\pi_e)$ as the sample size increases.

\begin{corollary}[Mean-square consistency]
\label{cor:adaptive_consistency}
Under the conditions of Theorem~\ref{thm:final-mse}, suppose that, as $n\to\infty$, $\lambda=\lambda_n$ satisfies $\lambda_n\to0$ and $n\lambda_n/\log n\to\infty$. Suppose that $\mathbb E[\delta_{H,K}^2]=o(\lambda_n^2)$ and $\mathbb E[\tau_{H,K}^2]=o(\lambda_n^4)$. Define $c_{H,K}=\min_{(s,a)\in\mathrm{supp}(d_H^{\pi_e})}q_{H,K}(s,a)$. If there exists a constant $c>0$ such that $\Pr(c_{H,K}\ge c)\to1$, then $\mathbb E[\{\widehat J(\widehat w_K)-J_H(\pi_e)\}^2]\to0$.
\end{corollary}

Since $\Psi_{H,\lambda_n}(q_{H,K})\le1/\lambda_n$, the condition $n\lambda_n/\log n\to\infty$ controls the sampling terms. If the collected data sufficiently cover the state--action pairs visited by the target policy, the effect of adding $\lambda_n$ to the denominator vanishes as $\lambda_n\to0$. It allows the offline logs alone to leave some target state--action pairs uncovered, provided that the combined sample achieves the stated coverage.

\endgroup

\subsection{Policy-Evaluation Error Bound for Continuous State Spaces}
\label{sec:continuous_extension}

\begingroup

For the continuous-state estimator in Section~\ref{sec:continuous_algorithm}, the superscript $\varphi$ distinguishes cell-level quantities from the state-level quantities in Section~\ref{sec:adaptive_plugin_bound}. For a policy $\pi$, define
\[
\begin{aligned}
d_H^{\pi,\varphi}(j,a)&=\int_{\mathcal S}\varphi_j(s)\,d_H^\pi(ds,a),\\
\bar r_H^{\pi,\varphi}(j,a)
&=\frac{\int_{\mathcal S}\varphi_j(s)r(s,a)\,d_H^\pi(ds,a)}{d_H^{\pi,\varphi}(j,a)},
\qquad d_H^{\pi,\varphi}(j,a)>0.
\end{aligned}
\]
For a state region $B\subseteq\mathcal S$, $d_H^\pi(B,a)$ denotes the normalized discounted visitation to $B$ with action $a$ over the first $H$ steps. We write $d_H^\pi(ds,a)$ when integrating over states. The quantity $d_H^{\pi,\varphi}(j,a)$ is the mass assigned to cell $j$ for action $a$, and $\bar r_H^{\pi,\varphi}(j,a)$ is the corresponding cell-weighted mean reward. Set $\bar r_H^{\pi,\varphi}(j,a)=0$ when $d_H^{\pi,\varphi}(j,a)=0$. The population sampling occupancy is
\[
q_{H,K}^{\varphi}(j,a)
=\frac{n_{\mathrm{off}}d_H^{\pi_b,\varphi}(j,a)
+\sum_{k=0}^{K-1}B_kd_H^{\tilde\mu_k,\varphi}(j,a)}{n},
\]
where $\tilde\mu_k$ is the policy executed by~\eqref{eq:continuous_probe_exec}. Define
\[
\begin{aligned}
\delta_{H,K}^{\varphi}
&=\sum_{j=1}^{K_\varphi}\sum_{a\in\mathcal A}
\left|\widehat d^{\pi_e,\varphi}_{H,K}(j,a)-d_H^{\pi_e,\varphi}(j,a)\right|,\\
\tau_{H,K}^{\varphi}
&=\max_{1\le j\le K_\varphi,\,a\in\mathcal A}
\left|\widehat q_K^{\varphi}(j,a)-q_{H,K}^{\varphi}(j,a)\right|.
\end{aligned}
\]
The estimates $\widehat d^{\pi_e,\varphi}_{H,K}$ and $\widehat q_K^{\varphi}$ are those in~\eqref{eq:continuous_mis_final}. Analogously to the definitions in Section~\ref{sec:adaptive_plugin_bound}, for a sampling occupancy $q$ on the cells, define
\[
\begin{aligned}
\Psi_{H,\lambda}^{\varphi}(q)
&:=\sum_{j=1}^{K_\varphi}\sum_{a\in\mathcal A}
\frac{d_H^{\pi_e,\varphi}(j,a)^2}{q(j,a)+\lambda},\\
\mathrm{Bias}_{H,\lambda}^{\varphi}(q)
&:=-\frac{1-\gamma^H}{1-\gamma}
\sum_{j=1}^{K_\varphi}\sum_{a\in\mathcal A}
\frac{\lambda\,d_H^{\pi_e,\varphi}(j,a)\bar r_H^{\pi_e,\varphi}(j,a)}{q(j,a)+\lambda}.
\end{aligned}
\]

\begin{corollary}[Continuous-state extension]
\label{cor:continuous_extension}
Consider a continuous state space, a finite action space, and a fixed aggregation map $\varphi$ specified before collecting the trajectories used for estimation. Suppose that the sampling and bounded-reward conditions of Theorem~\ref{thm:final-mse} hold, and that, with probability one, $\widehat d^{\pi_e,\varphi}_{H,K}(j,a)=0$ whenever $d_H^{\pi_e,\varphi}(j,a)=0$. Assume that, for some $\eta\ge0$, almost surely,
\[
\left|\bar r_H^{\pi,\varphi}(j,a)-\bar r_H^{\pi_e,\varphi}(j,a)\right|\le\eta,
\qquad
\pi\in\{\pi_b,\tilde\mu_0,\ldots,\tilde\mu_{K-1}\},
\]
whenever $d_H^{\pi_e,\varphi}(j,a)>0$ and $d_H^{\pi,\varphi}(j,a)>0$.

Let $m=|\{(j,a):1\le j\le K_\varphi,\ a\in\mathcal A,\ d_H^{\pi_e,\varphi}(j,a)>0\}|$ be the number of cell--action pairs visited by the target policy. For $m\ge1$, $\lambda>0$, and $\zeta\in(0,1)$, set $\ell=\log\{2\mathrm e\,m[1+\lceil\log_2 n\rceil]/\zeta\}$ as in Theorem~\ref{thm:final-mse}. The estimator in~\eqref{eq:continuous_mis_final} then satisfies
\[
\begin{aligned}
&\mathbb E\!\left[\bigl(\widehat J(\widehat w_K^{\varphi})-J_H(\pi_e)\bigr)^2\right]
\le R_{\max}^2\left(\frac{1-\gamma^H}{1-\gamma}\right)^2
\Biggl\{\frac{24m\ell}{n}\,\mathbb E[\Psi_{H,\lambda}^{\varphi}(q_{H,K}^{\varphi})]\\
&\qquad+\frac{96\ell^2}{n^2\lambda^2}
+\zeta(1+1/\lambda)^2
+3\mathbb E\!\left[\left(\frac{\delta_{H,K}^{\varphi}}{\lambda}
+\frac{\tau_{H,K}^{\varphi}}{\lambda^2}\right)^2\right]\Biggr\}\\
&\qquad+6\mathbb E\bigl[\mathrm{Bias}_{H,\lambda}^{\varphi}(q_{H,K}^{\varphi})^2\bigr]
+6\left(\frac{1-\gamma^H}{1-\gamma}\right)^2\eta^2.
\end{aligned}
\]

For $m\ge1$, let $c_{H,K}^{\varphi}=\min_{(j,a):\,d_H^{\pi_e,\varphi}(j,a)>0}q_{H,K}^{\varphi}(j,a)$. If $\lambda_n\to0$, $n\lambda_n/\log n\to\infty$, $\mathbb E[(\delta_{H,K}^{\varphi})^2]=o(\lambda_n^2)$, $\mathbb E[(\tau_{H,K}^{\varphi})^2]=o(\lambda_n^4)$, and $\Pr(c_{H,K}^{\varphi}\ge c)\to1$ for some $c>0$, then
\[
\limsup_{n\to\infty}\mathbb E\!\left[\{\widehat J(\widehat w_K^{\varphi})-J_H(\pi_e)\}^2\right]\le6\left(\frac{1-\gamma^H}{1-\gamma}\right)^2\eta^2.
\]
If the reward differences admit bounds $\eta=\eta_n\to0$ as $n\to\infty$, these conditions imply $\mathbb E[\{\widehat J(\widehat w_K^{\varphi})-J_H(\pi_e)\}^2]\to0$.
\end{corollary}

The fixed-aggregation condition holds for a prespecified map or one fitted using independent pilot data excluded from the value estimator. Algorithm~\ref{alg:cad_continuous} instead constructs the aggregation from the operational logs and reuses those logs for estimation. Corollary~\ref{cor:continuous_extension} does not account for this dependence; extending the bound to this data reuse requires further analysis.

The bound highlights a trade-off in the choice of state aggregation. Finer cells may reduce differences in average rewards, but can be harder to cover and estimate accurately with the available data. A useful aggregation therefore groups states with similar rewards for each action while retaining enough observations to estimate cell occupancies.

\endgroup

\section{Numerical Study}
\label{sec:simulation}

We evaluate the design framework in two simulated environments: Gridworld and LunarLander, the latter with a continuous state space. Both experiments measure end-to-end performance across a range of offline log sizes $n_{\mathrm{off}}$ and audit budgets $B$.
\begin{figure}[htbp]
  \centering
  \begin{subfigure}[b]{0.49\textwidth}
    \centering
    \includegraphics[width=\linewidth,height=5.6cm,keepaspectratio]{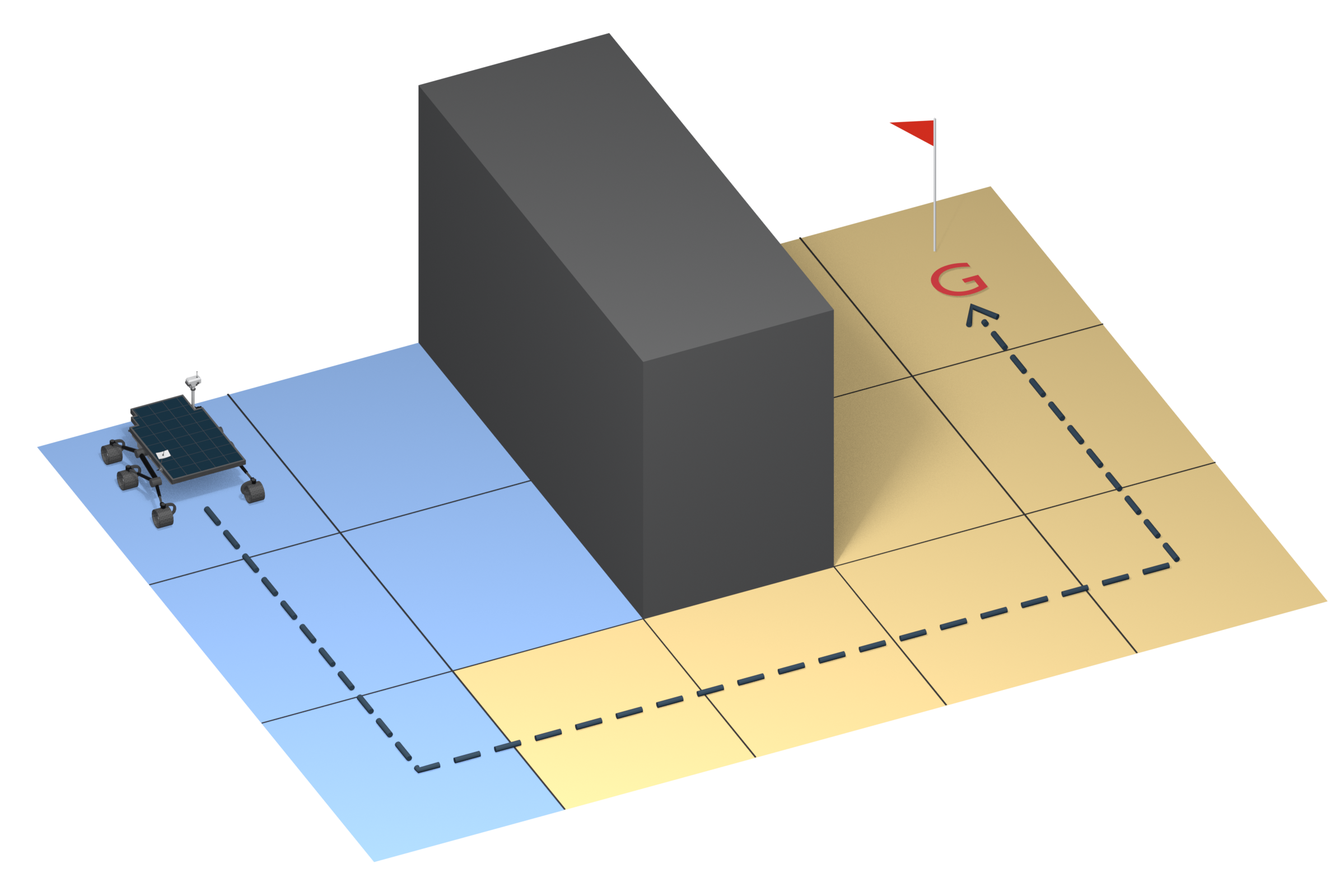}
    \caption{}
    \label{fig:setting_3d}
  \end{subfigure}
  \hspace{0.01\textwidth}%
  \begin{subfigure}[b]{0.49\textwidth}
    \centering
    \includegraphics[width=\linewidth,height=5.1cm,keepaspectratio]{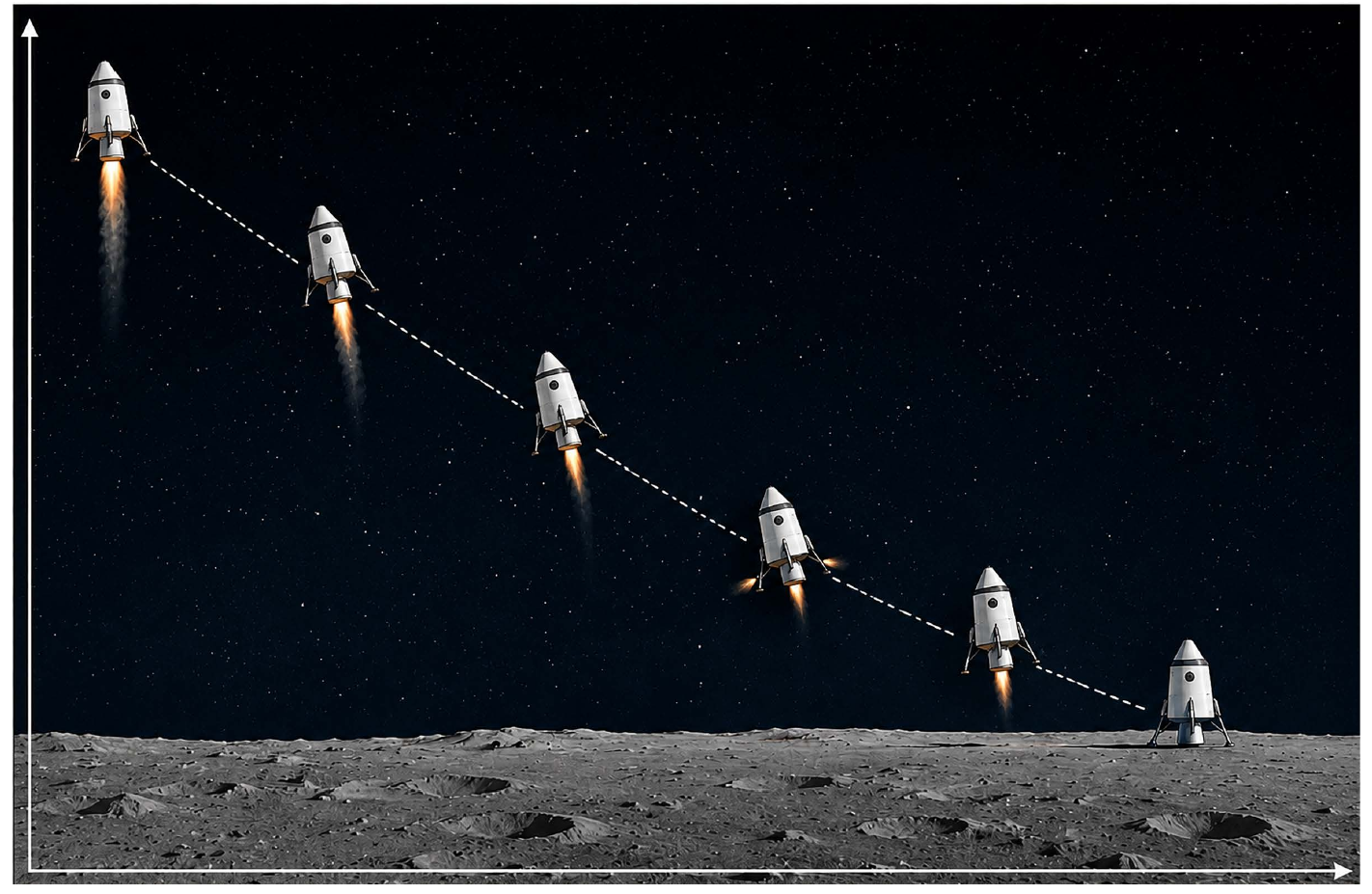}
    \caption{}
    \label{fig:lunarlander_setting}
  \end{subfigure}
  \caption{Simulation environments. (a) The $5\times 3$ Gridworld: the agent must detour around a wall to reach the goal $G$. (b) LunarLander-v3: a rocket learns to land safely on the ground.}
  \label{fig:setting}
\end{figure}

\subsection{Policy Evaluation in Gridworld}
\label{sec:gridworld_setup}
We consider a $5\times 3$ Gridworld task with $|\mathcal{S}|=15$ states, $|\mathcal{A}|=4$ actions, discount $\gamma=0.97$, and horizon $H=20$ (Figure~\ref{fig:setting_3d}). A wall blocks the direct horizontal path from the upper-left start to the upper-right goal, forcing a detour through the bottom row. The environment executes the selected action with probability $0.9$ and samples an action uniformly from the four available actions with probability $0.1$.    The goal is an absorbing state: after reaching it, the agent remains there and receives reward $+1$ at every remaining step of the horizon.    In all other states, it receives $-0.01$ per step.

The behavior policy $\pi_b$ selects an action uniformly at random with probability $0.7$ and selects left with probability $0.3$. It mainly visits the left columns, shown in blue in Figure~\ref{fig:setting}(a), and rarely reaches the bottom row or the region right of the wall. The target policy $\pi_e$ uses a softmax rule to favor actions toward the goal. It frequently visits regions poorly covered by the behavior policy, creating an \emph{audit blind spot} where the offline logs provide little information for evaluating $\pi_e$.

In the end-to-end experiment, the transition kernel $P$ and all occupancy measures are estimated from the data collected in each round (Algorithm~\ref{alg:probe_design}). The target-policy value $J_H(\pi_e)=3.016$ provides the reference for measuring estimation error.    We use $K=3$ audit rounds with the budget allocation in Section~\ref{sec:algorithm} and $\lambda=\varepsilon=10^{-6}$. The offline dataset size is $n_{\mathrm{off}}\in\{100,200,500\}$, the audit budget is $B\in\{5,10,20,50,100,200\}$, and all metrics are averaged over $50$ independent replications.

We compare \textbf{MIS} \citep{liu2018breaking}, \textbf{random probe} (uniform random audit actions combined with offline data), \textbf{target probe} (auditing with $\pi_e$ itself), and \textbf{CAD}.     Two additional offline baselines, per-decision importance sampling (\textbf{PDIS}) \citep{precup2000eligibility} and doubly robust estimation (\textbf{DR}) \citep{jiang2016doubly,dudik2014doubly}, use the same operational logs.

Table~\ref{tab:endtoend} reports the MSE relative to $J_H(\pi_e)$, estimated from $50$ replications in each of the eighteen $(n_{\mathrm{off}},B)$ settings.    Panel (a) reports the offline MIS, PDIS, and DR results, while panel (b) reports the probe methods across audit budgets. Neither PDIS nor DR consistently improves on MIS across the three offline sample sizes.        The target probe improves coverage, but CAD achieves lower MSE at small audit budgets. As the budget increases, both methods yield low estimation error. These results show that probe design can substantially reduce estimation error.    CAD steers probes toward poorly covered pairs while accounting for the fixed offline logs.

\begin{table}[htbp]
  \centering
  \caption{End-to-end MSE on the $5\times 3$ Gridworld ($50$ replications; $\lambda=\varepsilon=10^{-6}$), target $J_H(\pi_e)=3.016$. Panel (a) reports offline baselines using only the operational logs; panel (b) reports methods that add $B$ probe trajectories. Bold marks the lowest MSE among probe methods at each $(n_{\mathrm{off}}, B)$.}
  \label{tab:endtoend}
  \small
  \setlength{\tabcolsep}{4pt}
  \renewcommand{\arraystretch}{1.15}
  \begin{tabular}{@{}c @{\hspace{1.2em}} ccc @{\hspace{1.2em}} ccc @{\hspace{1.2em}} ccc@{}}
    \multicolumn{10}{@{}l}{\textit{(a) Offline evaluation}} \\[3pt]
    \toprule
    & \multicolumn{3}{c}{$n_{\mathrm{off}} = 100$} & \multicolumn{3}{c}{$n_{\mathrm{off}} = 200$} & \multicolumn{3}{c}{$n_{\mathrm{off}} = 500$} \\
    \cmidrule(lr){2-4} \cmidrule(lr){5-7} \cmidrule(lr){8-10}
    & MIS & PDIS & DR & MIS & PDIS & DR & MIS & PDIS & DR \\
    & 8.812 & 9.696 & 11.768 & 8.018 & 1661.480 & 7.065 & 5.668 & 293.858 & 6.880 \\
    \bottomrule
    \addlinespace[1.5ex]
    \multicolumn{10}{@{}l}{\textit{(b) Evaluation with probes}} \\[3pt]
    \toprule
    & \multicolumn{3}{c}{$n_{\mathrm{off}} = 100$} & \multicolumn{3}{c}{$n_{\mathrm{off}} = 200$} & \multicolumn{3}{c}{$n_{\mathrm{off}} = 500$} \\
    \cmidrule(lr){2-4} \cmidrule(lr){5-7} \cmidrule(lr){8-10}
    $B$ & Rand. & Targ. & CAD & Rand. & Targ. & CAD & Rand. & Targ. & CAD \\
    \midrule
      5 & 8.494 & 0.744 & \textbf{0.085} & 7.375 & 0.425 & \textbf{0.120} & 6.629 & 0.257 & \textbf{0.093} \\
     10 & 8.108 & 0.212 & \textbf{0.068} & 7.600 & 0.171 & \textbf{0.053} & 6.461 & 0.114 & \textbf{0.033} \\
     20 & 7.411 & 0.062 & \textbf{0.040} & 7.182 & 0.056 & \textbf{0.028} & 3.221 & 0.033 & \textbf{0.017} \\
     50 & 3.866 & 0.019 & \textbf{0.016} & 3.670 & 0.019 & \textbf{0.017} & 3.255 & 0.017 & \textbf{0.012} \\
    100 & 2.474 & 0.013 & \textbf{0.011} & 1.814 & \textbf{0.006} & 0.008 & 0.724 & 0.008 & \textbf{0.006} \\
    200 & 0.431 & 0.006 & \textbf{0.004} & 0.542 & 0.004 & \textbf{0.004} & 0.248 & 0.004 & \textbf{0.004} \\
    \bottomrule
  \end{tabular}
\end{table}

\subsection{Policy Evaluation in LunarLander}
\label{sec:lunarlander}

To evaluate CAD beyond the tabular setting, we apply it to \textbf{LunarLander-v3}~\citep{gymnasium}, a continuous state control task with an $8$-dimensional state space and $|\mathcal{A}|=4$ discrete actions (Figure~\ref{fig:lunarlander_setting}). The eight state components are horizontal and vertical position, horizontal and vertical velocity, orientation, angular velocity, and two indicators of ground contact for the landing legs. The four actions are to fire no engine, the left side engine, the main engine, or the right side engine.     We use the reward function built into LunarLander-v3, which favors approaching the landing pad slowly and upright, penalizes engine use, and assigns terminal rewards of $+100$ for a safe landing and $-100$ for a crash.    The behavior policy $\pi_b$ is a landing policy pretrained using proximal policy optimization (PPO).    The target policy follows the same action distribution, except that below altitude $1.0$ it replaces the action by firing the left side engine with probability $0.2$. This perturbation induces lateral drift and brings the target into low-altitude states poorly represented in the operational logs. We use $\gamma=0.99$ and $H=1000$. An independent pool of $20{,}000$ target-policy episodes gives a Monte Carlo reference for the expected discounted return, $J_H(\pi_e)\approx14$.
We implement Algorithm~\ref{alg:cad_continuous} with $K_\varphi=160$ cells, $m=4$ nearest cluster centers, and $K=3$ audit rounds. Features are standardized on this pool, and the abstraction remains fixed across rounds. We train $M=5$ transition models, using the multilayer perceptron (MLP) structure and Adam~\citep{kingma2015adam} algorithm. Each model generates $N_{\mathrm{sim}}=500$ target-policy episodes. All simulated episodes enter the target-occupancy estimate and the pooled transition counts.

We compare MIS, random probe, target probe, and CAD. Here, the random probe follows $\pi_b$ with probability $0.8$ and chooses a uniformly random action otherwise, while the target probe executes $\pi_e$. We use $n_{\mathrm{off}}\in\{1000,2000,5000\}$ and $B\in\{10,20,50,100,200\}$, with $30$ replications per setting. We set $\lambda=\varepsilon=10^{-6}/(1-\gamma^H)$.

Table~\ref{tab:lunarlander} shows that random probing provides little improvement over offline MIS. CAD consistently achieves lower MSE than random and target probing and improves on PDIS and DR in most settings.
\begin{table}[htbp]
  \centering
  \caption{MSE on LunarLander ($K_\varphi=160$, $K=3$ rounds, and $\lambda=\varepsilon=10^{-6}/(1-\gamma^H)$). Panel (a) reports offline baselines; panel (b) reports methods that add $B$ probe trajectories. Bold marks the lowest MSE among probe methods at each $(n_{\mathrm{off}}, B)$.}
  \label{tab:lunarlander}
  \small
  \setlength{\tabcolsep}{4pt}
  \renewcommand{\arraystretch}{1.15}
  \begin{tabular}{@{}c @{\hspace{1.2em}} ccc @{\hspace{1.2em}} ccc @{\hspace{1.2em}} ccc@{}}
    \multicolumn{10}{@{}l}{\textit{(a) Offline evaluation}} \\[3pt]
    \toprule
    & \multicolumn{3}{c}{$n_{\mathrm{off}} = 1000$} & \multicolumn{3}{c}{$n_{\mathrm{off}} = 2000$} & \multicolumn{3}{c}{$n_{\mathrm{off}} = 5000$} \\
    \cmidrule(lr){2-4} \cmidrule(lr){5-7} \cmidrule(lr){8-10}
    & MIS & PDIS & DR & MIS & PDIS & DR & MIS & PDIS & DR \\
    & 120.68 & 109.69 & 134.84 & 123.59 & 164.15 & 131.51 & 155.73 & 102.74 & 66.52 \\
    \bottomrule
    \addlinespace[1.5ex]
    \multicolumn{10}{@{}l}{\textit{(b) Evaluation with probes}} \\[3pt]
    \toprule
    & \multicolumn{3}{c}{$n_{\mathrm{off}} = 1000$} & \multicolumn{3}{c}{$n_{\mathrm{off}} = 2000$} & \multicolumn{3}{c}{$n_{\mathrm{off}} = 5000$} \\
    \cmidrule(lr){2-4} \cmidrule(lr){5-7} \cmidrule(lr){8-10}
    $B$ & Rand. & Targ. & CAD & Rand. & Targ. & CAD & Rand. & Targ. & CAD \\
    \midrule
     10 & 137.84 & 112.63 & \textbf{ 63.41} & 122.20 & 109.41 & \textbf{ 79.81} & 158.51 & 155.84 & 70.15 \\
     20 & 133.45 &  90.94 & \textbf{ 32.05} & 131.04 & 105.23 & \textbf{ 42.04} & 157.12 & 138.53 & \textbf{ 62.85} \\
     50 & 162.29 &  74.91 & \textbf{ 28.05} & 146.83 &  75.60 & \textbf{ 35.04} & 145.62 &  91.80 & \textbf{ 33.73} \\
    100 & 154.04 &  48.59 & \textbf{ 30.70} & 140.22 &  66.84 & \textbf{ 22.16} & 146.55 &  82.77 & \textbf{ 33.54} \\
    200 & 184.79 &  55.46 & \textbf{ 23.75} & 126.95 &  57.62 & \textbf{  6.43} & 147.54 &  58.25 & \textbf{ 19.75} \\
    \bottomrule
  \end{tabular}
\end{table}

\section{Real-World Case Study: Autonomous Driving}
\label{sec:driving}

Our real-world study audits an autonomous driving policy against real traffic (Figure~\ref{fig:driving_intro}). We work inside Waymax, a hardware-accelerated driving simulator, running the V-Max reinforcement learning pipeline over logs from the Waymo Open Motion Dataset \citep{ettinger2021waymo, gulino2023waymax}.
Each vehicle carries five LiDARs and five cameras. The dataset gives, for every tracked object, its position, heading, and velocity at each time step, together with the surrounding road geometry.    Positions and velocities are expressed in the world coordinate system. Headings are measured in radians, and object dimensions in meters.  Figure~\ref{fig:driving_intro} illustrates selected logged fields for an ego vehicle and a pedestrian. The dataset provides approximately $450{,}000$ scenarios for model training. A scenario spans $91$ frames taken a tenth of a second apart, about nine seconds of traffic.    We evaluate the policies on scenarios excluded from target-policy training. Each rollout runs for up to $H=80$ steps, with $\gamma=0.9$.

\begin{figure}[!htb]
\centering
\includegraphics[width=0.95\textwidth]{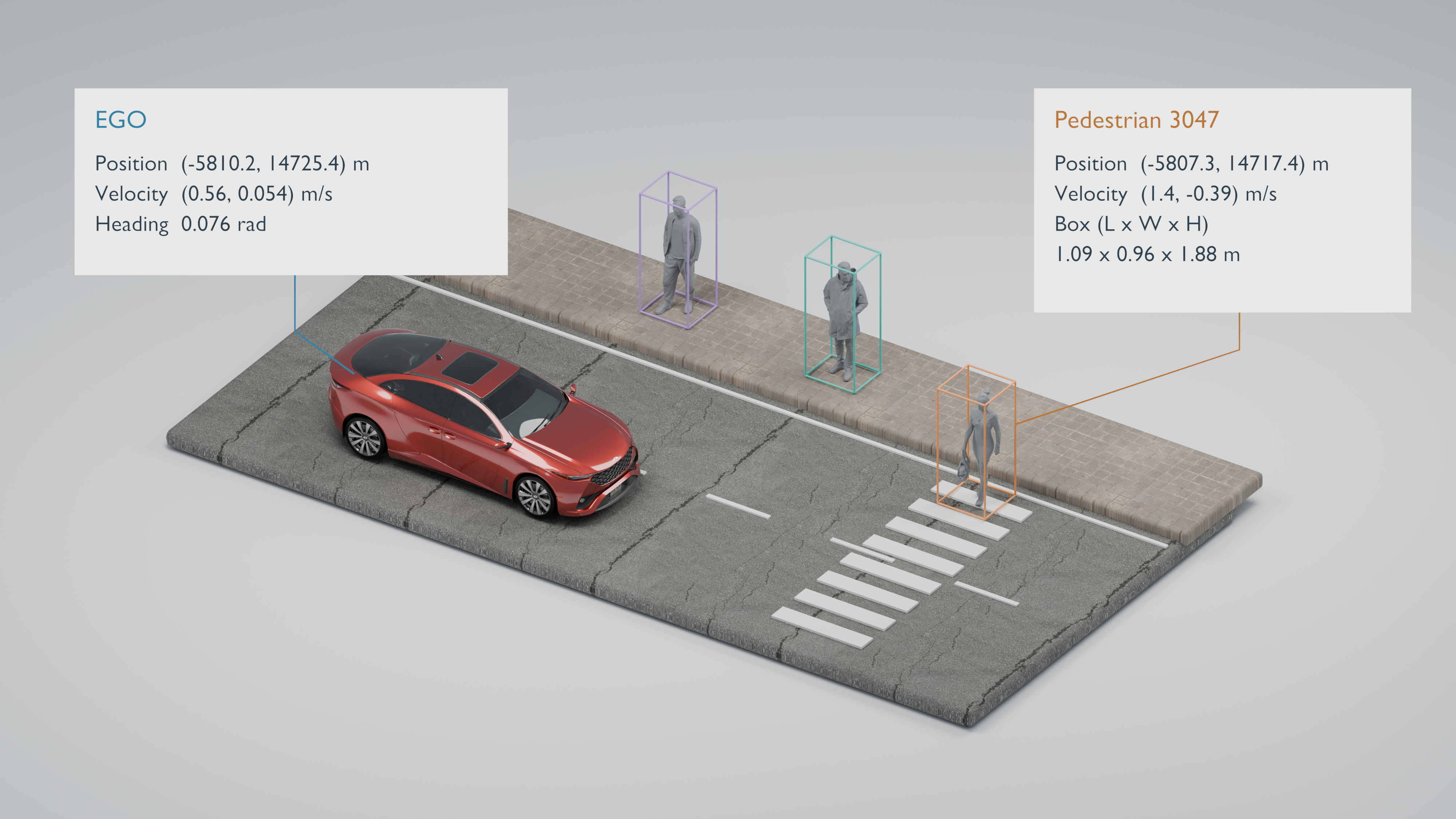}
\caption{Illustrative driving scene with selected WOMD annotations for the ego vehicle (left label) and a pedestrian (numbered 3047, right label, orange box). Positions and velocity components are in world coordinates; the ego vehicle's heading is in radians, and the pedestrian's box dimensions are length, width, and height in meters.    The road layout, object models, and displayed boxes are schematic.}
\label{fig:driving_intro}
\end{figure}

The simulator state includes the ego vehicle's motion and the surrounding traffic and road geometry. For clustering, we extract $12$ ego-relative features. An action consists of acceleration and steering curvature, the continuous controls of the bicycle model \citep{polack2017kinematic} used by Waymax. The reward is the simulator's per-step safety penalty for leaving the road or overlapping another object.

The behavior policy $\pi_b$ follows the logged human driver, with controls recovered by inverting the bicycle model. The target policy $\pi_e$ is trained using soft actor-critic \citep{haarnoja2018soft} in V-Max. Its value is negative because the reward is a safety penalty.

Soft $K$-means \citep{bezdek1981pattern, li2006towards} groups the standardized features into $K_{\mathrm{soft}}=40$ clusters with four weights per state. Following Algorithm~\ref{alg:cad_continuous}, the clustering pool combines offline states with model-simulated target states and is fixed across rounds. Actions are assigned to the nearest point on a $5\times5$ grid whose coordinates span the $1$st to $99$th percentiles of the offline controls, giving $25$ action cells.
In each round, we fit $M=5$ MLP as the transition models of the ego vehicle's dynamics to the accumulated real trajectories. Each model simulates $N_{\mathrm{sim}}=200$ target-policy episodes to estimate the target occupancy and the pooled cell transitions. We use $K=3$ rounds and $\lambda=\varepsilon=10^{-6}/(1-\gamma^H)$.

We compare MIS, PDIS, DR, random probe, target probe, and CAD. All methods in a replication share the offline logs, and all probe methods use the same $B$ scenarios following the offline block. The random probe follows the logged driver with probability $0.8$ and otherwise samples a control uniformly within a uniformly selected action cell. MIS and the probe methods estimate their target occupancies from models fitted to their own real trajectories.

\begin{table}[!ht]
\centering
\caption{RMSE ($\times 10^{2}$) on the WOMD autonomous-driving audit across seven $(n_{\mathrm{off}},B)$ settings ($50$ replications each), with bootstrap standard errors in parentheses. We use $K_{\mathrm{soft}}=40$, $K=3$ rounds, and $\lambda=\varepsilon=10^{-6}/(1-\gamma^H)$. Bold marks the lowest RMSE in each row. PDIS and DR are omitted because of numerical divergence.}
\label{tab:driving_scale}
\small
\setlength{\tabcolsep}{6pt}
\begin{tabular}{rr cccc}
  \toprule
  $n_{\mathrm{off}}$ & $B$ & MIS & Random probe & Target probe & CAD \\
  \midrule
  1000 &  10 & 3.82\,(0.05) & 3.57\,(0.08) & 3.11\,(0.11) & \textbf{2.48\,(0.14)} \\
  1000 &  20 & 3.80\,(0.05) & 3.41\,(0.10) & 2.87\,(0.11) & \textbf{1.88\,(0.20)} \\
  1000 &  50 & 3.82\,(0.05) & 3.14\,(0.11) & 2.34\,(0.13) & \textbf{1.13\,(0.09)} \\
  1000 & 100 & 3.83\,(0.05) & 2.85\,(0.11) & 1.88\,(0.13) & \textbf{1.43\,(0.16)} \\
  1000 & 200 & 3.86\,(0.03) & 2.51\,(0.11) & 1.44\,(0.10) & \textbf{1.39\,(0.13)} \\
  \addlinespace
  2000 & 100 & 3.77\,(0.04) & 3.00\,(0.10) & 2.07\,(0.11) & \textbf{1.13\,(0.09)} \\
  \addlinespace
   500 & 100 & 4.03\,(0.02) & 2.88\,(0.13) & \textbf{1.55\,(0.12)} & 1.72\,(0.18) \\
  \bottomrule
\end{tabular}
\end{table}

For each configuration, we approximate the target-policy value $J_H(\pi_e)$ by independently running the target policy once on each evaluation scenario and averaging the discounted returns. These additional trajectories are used only to evaluate estimation error. Table~\ref{tab:driving_scale} reports RMSE for seven $(n_{\mathrm{off}},B)$ configurations, each with $50$ replications; parentheses give bootstrap standard errors. PDIS and DR yield numerically unstable estimates, with RMSE greater than $10$ in every configuration, and are therefore omitted from the table.

Figure~\ref{fig:driving_heatmap} compares the occupancies of the offline logs, the target policy, and the collected CAD probes. The highlighted region illustrates how CAD directs additional visits toward state--action cells frequently visited by the target policy but poorly covered by the logs.

\begin{figure}[htbp]
\centering
\includegraphics[width=\textwidth]{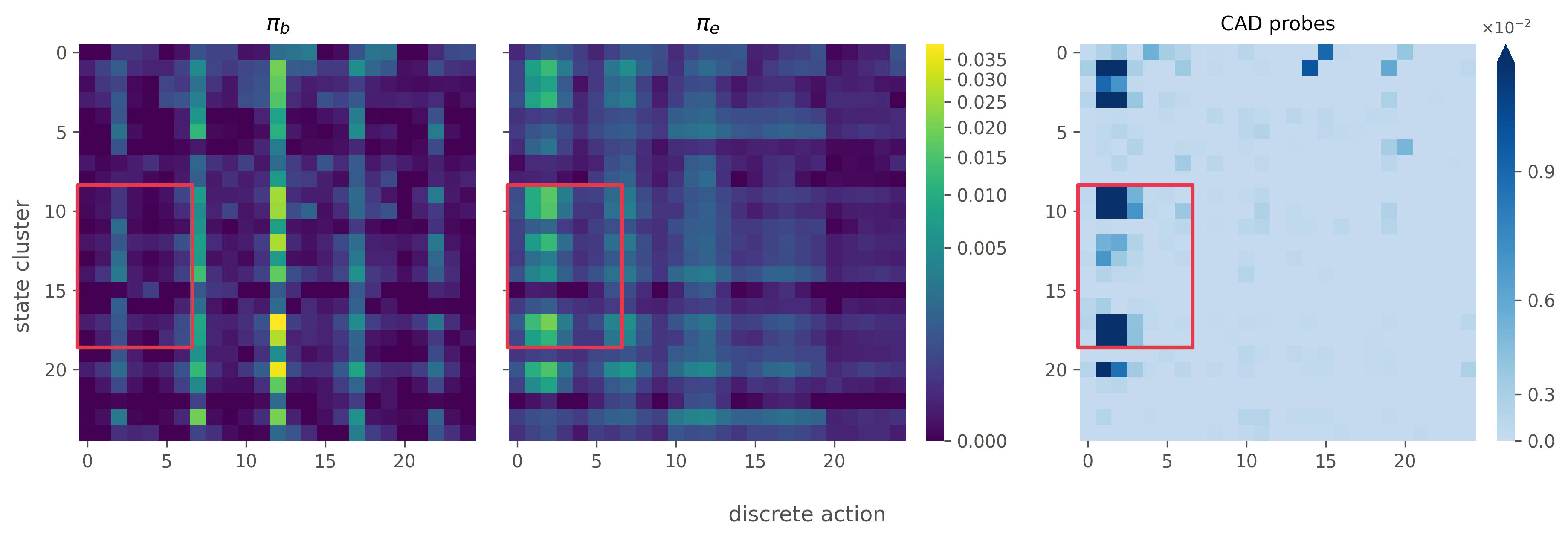}
\caption{Empirical state--action occupancies in the autonomous driving study ($n_{\mathrm{off}}=1000$, $B=100$). Left: offline logs. Center: independent target-policy rollouts. Right: collected CAD probe trajectories across all rounds. We show all $25$ action cells and the contiguous $25$-cluster window with the largest probe occupancy among the $40$ clusters. The red boxes highlight an audit blind spot: cells frequently visited by the target policy but poorly covered by the offline logs.
CAD probes improve coverage in this region.}
\label{fig:driving_heatmap}
\end{figure}

\FloatBarrier
\section{Conclusion}
\label{sec:discussion}

We study the problem of auditing AI agents using fixed operational logs and a limited budget of probe trajectories.    CAD selects dynamically feasible probe policies by minimizing an occupancy-mismatch criterion and updates its design as new data arrive. CAD directs probes toward state--action pairs that are frequently visited by the target policy but poorly represented in the operational logs. Our theoretical analysis links this criterion to estimation-error bounds and accounts for adaptive sampling, estimated occupancies, and state aggregation. Experiments in Gridworld, LunarLander, and autonomous driving demonstrate substantial gains in policy-value estimation accuracy when operational logs provide limited target-policy coverage, including orders-of-magnitude reductions in mean squared error relative to offline MIS in Gridworld. These results demonstrate that CAD can improve policy-value estimation with a limited number of additional audit trajectories. This is particularly useful when collecting new trajectories is costly and existing logs provide insufficient coverage of the target policy.        

Future work includes learning state abstractions for higher-dimensional environments and designing probes for lower-tail return estimation.     The latter would extend CAD to risk-sensitive auditing, motivated by quantile-based evaluation~\citep{wang2025quantile}, and requires probe objectives tailored to tail estimation.

\begingroup
\setstretch{1}
\setlength{\bibsep}{4pt}
\bibliography{reference}
\endgroup

\clearpage
\phantomsection
\addcontentsline{toc}{section}{Supplementary Material}
\begin{center}
  {\Large\bfseries Supplementary Material}
\end{center}
\appendix
\numberwithin{equation}{section}
\section{Additional Results for Probe Design}
\label{app:design_results}

\begingroup

\subsection{Limits of Probe Design}
\label{app:design_limits}
\label{sec:design_gap}

We consider the infinite-horizon design criterion of Section~\ref{sec:problem} and characterize when the smallest achievable mismatch attains its lower bound. Although the probe policy can be chosen freely, the effective occupancy must retain the contribution $(1-\alpha)d^{\pi_b}$ from the offline logs.

The Bellman flow-feasible set is
\[
\mathcal D_{\mathrm{flow}}
=\bigl\{x\ge0 :
\sum_{a\in\mathcal A}x(s,a)
=(1-\gamma)d_0(s)
+\gamma\sum_{s'\in\mathcal S,a'\in\mathcal A}x(s',a')P(s\mid s',a'),
\quad\forall s\in\mathcal S\bigr\}.
\]

\begin{theorem}[Limits under fixed offline logs]
\label{thm:flow_reach_gap}
For $0<\alpha<1$, define the feasible set of effective occupancies as
\[
\mathcal Q_\alpha
=\bigl\{(1-\alpha)d^{\pi_b}+\alpha x : x\in\mathcal D_{\mathrm{flow}}\bigr\},
\]
and let $\Psi^\star=\min_{q\in\mathcal Q_\alpha}\Psi(q)$. Then
$
1\le\Psi^\star\le\frac{1}{\alpha}.
$
Equality in the lower bound is characterized by
\[
\Psi^\star=1
\quad\Longleftrightarrow\quad
 d^{\pi_e}(s,a)\ge(1-\alpha)d^{\pi_b}(s,a)
\quad\text{for all }(s,a).
\]
When this condition holds, the probe occupancy $x^\star=\{d^{\pi_e}-(1-\alpha)d^{\pi_b}\}/\alpha$ is flow-feasible: it is nonnegative and satisfies the Bellman flow constraints defining $\mathcal D_{\mathrm{flow}}$. The corresponding probe policy $\mu^\star$ yields $q_\alpha^{\mu^\star}=d^{\pi_e}$. Otherwise, $\Psi^\star>1$.
\end{theorem}

The proof is in Appendix~\ref{app:design_limits_proof}. The lower bound follows from $\Psi(q)=1+\chi^2(d^{\pi_e}\|q)$, and the upper bound follows by choosing the target policy as the probe. The equality condition identifies the obstruction to exact matching: any offline contribution exceeding the target occupancy cannot be removed by adding probe trajectories. If no coordinate has this excess, the displayed $x^\star$ is nonnegative and satisfies the Bellman flow equations because $d^{\pi_e}$ and $d^{\pi_b}$ share the same dynamics and initial distribution. Thus, the strict gap characterized here arises from retaining the offline contribution. For fixed $n_{\mathrm{off}}$, increasing $B$ raises $\alpha$ and brings the upper bound $1/\alpha$ toward one.

Flow constraints still determine the optimal allocation when exact matching is infeasible. Proposition~\ref{prop:kkt_flow} in Appendix~\ref{app:kkt} characterizes an optimal occupancy through the KKT conditions, which balance squared importance ratios against the dual variables associated with the Bellman flow equations. These conditions account for how allocating probe occupancy at one state--action pair affects coverage elsewhere through the transition dynamics. Without flow constraints, the optimal allocation equalizes importance weights across pairs receiving positive probe allocation, analogous to Neyman allocation in survey sampling \citep{neyman1992two}; see Remark~\ref{rem:gm_contrast} in Appendix~\ref{app:kkt}.

For a small audit fraction $\alpha$, the same design problem can be interpreted through a local expansion around the offline occupancy $d^{\pi_b}$. For this expansion, we restrict attention to $\mathrm{supp}(d^{\pi_e})\subseteq\mathrm{supp}(d^{\pi_b})$; otherwise, the baseline mismatch $\Psi(d^{\pi_b})$ is infinite. To first order in $\alpha$, minimizing $\Psi$ is equivalent to maximizing the expected discounted sum of an auxiliary reward equal to $(d^{\pi_e}/d^{\pi_b})^2$ on the target support and zero elsewhere. The probe therefore seeks regions with large squared importance ratios while respecting the dynamics needed to reach them. Appendix~\ref{app:probing_gain} gives the exact mismatch decomposition and the first-order expansion underlying this interpretation.
\endgroup

\subsection{KKT Characterization of the Optimal Probe Occupancy}
\label{app:kkt}

\begin{proposition}[KKT characterization of the flow-constrained optimal probe occupancy]
\label{prop:kkt_flow}
For $0<\alpha<1$, let
\[
x^*\in\argmin_{x\in\mathcal D_{\mathrm{flow}}}
\Psi\!\left((1-\alpha)d^{\pi_b}+\alpha x\right),
\]
and let $q^* = (1-\alpha)d^{\pi_b} + \alpha x^*$. Set $d^{\pi_e}(s,a)/q^*(s,a)=0$ whenever $d^{\pi_e}(s,a)=0$, including when $q^*(s,a)=0$. There exist dual variables $\{\lambda_s\}_{s\in\mathcal{S}}$ associated with the Bellman flow constraints such that, for each $(s,a)\in\mathcal{S}\times\mathcal{A}$,
\[
\alpha\left(\frac{d^{\pi_e}(s,a)}{q^*(s,a)}\right)^{\!2}
\;\leq\;
\lambda_s - \gamma\sum_{s'' \in \mathcal{S}}\lambda_{s''}\,P(s''\mid s,a),
\]
with equality when $x^*(s,a)>0$, and $x^*(s,a)=0$ whenever
\[
\alpha\left(\frac{d^{\pi_e}(s,a)}{q^*(s,a)}\right)^{\!2}
\;<\;
\lambda_s - \gamma\sum_{s'' \in \mathcal{S}}\lambda_{s''}\,P(s''\mid s,a).
\]
\end{proposition}

\begin{remark}[Contrast with the unconstrained generative model]
\label{rem:gm_contrast}
In the absence of flow constraints, where the probe occupancy can be freely chosen from the full simplex $\Delta(\mathcal{S}\times\mathcal{A})$, the optimal probe admits a closed form.
The unique minimizer $x^*_{\mathrm{gm}}\in\arg\min_{x\in\Delta(\mathcal{S}\times\mathcal{A})} \Psi\bigl((1-\alpha)d^{\pi_b}+\alpha x\bigr)$ satisfies
\[
x^*_{\mathrm{gm}}(s,a) \;=\; \frac{1}{\alpha}\bigl(\tau\,d^{\pi_e}(s,a)-(1-\alpha)\,d^{\pi_b}(s,a)\bigr)_+,
\]
where $(\cdot)_+:=\max\{\cdot,0\}$ and $\tau>0$ is the unique constant such that $\sum_{s \in \mathcal{S},\, a \in \mathcal{A}}x^*_{\mathrm{gm}}(s,a)=1$.
The corresponding importance weight is constant across all active pairs: $w^*(s,a)=d^{\pi_e}(s,a)/q^*(s,a)=1/\tau$ for all $(s,a)$ with $x^*_{\mathrm{gm}}(s,a)>0$, a direct analogue of Neyman allocation in survey sampling \citep{neyman1992two}.

In the flow-constrained setting, the dual variables $\{\lambda_s\}$ from Proposition~\ref{prop:kkt_flow} are coupled through the transition kernel $P$. The KKT conditions involve $\lambda_s-\gamma\sum_{s'}\lambda_{s'}P(s'\mid s,a)$, so the optimal allocation depends on both occupancy mismatch and the Bellman flow constraints. In the generative model, importance weights are equalized across pairs receiving positive probe allocation.
\end{remark}

The proof is in Appendix~\ref{app:kkt_proof}.

\subsection{Occupancy-Mismatch Reduction: Exact Decomposition and Small-Budget Expansion}
\label{app:probing_gain}

This subsection derives the local design interpretation in Appendix~\ref{app:design_limits}. Throughout, the MDP is finite and $\mathcal T=\mathrm{supp}(d^{\pi_e})\subseteq\mathrm{supp}(d^{\pi_b})$. We restrict ratio expressions to $\mathcal T$ so that their denominators are positive. The results concern changes in the mismatch criterion; an arbitrary probe need not reduce it.

\begin{proposition}[Exact decomposition of probing gain]
\label{thm:exact_decomp}
For any flow-feasible probe occupancy $x=d^\mu$ and $0\le\alpha<1$,
\[
\Psi(d^{\pi_b})-\Psi(q_\alpha^\mu)
=\alpha\sum_{(s,a)\in\mathcal T}
\frac{d^{\pi_e}(s,a)^2\bigl(x(s,a)-d^{\pi_b}(s,a)\bigr)}
{d^{\pi_b}(s,a)\,q_\alpha^\mu(s,a)}.
\]
\end{proposition}

\begin{proposition}[First-order expansion under a small audit fraction]
\label{thm:first_order}
As $\alpha\downarrow0$,
\[
\Psi(d^{\pi_b})-\Psi(q_\alpha^\mu)
=\alpha\sum_{(s,a)\in\mathcal T}
\left(\frac{d^{\pi_e}(s,a)}{d^{\pi_b}(s,a)}\right)^2
\bigl(x(s,a)-d^{\pi_b}(s,a)\bigr)
+O(\alpha^2),
\]
uniformly over $x\in\mathcal D_{\mathrm{flow}}$.
\end{proposition}

The uniform remainder permits optimization over the feasible occupancies. With $\Psi^\star$ evaluated at audit fraction $\alpha$,
\[
\Psi(d^{\pi_b})-\Psi^\star
=\alpha\left[
\max_{x\in\mathcal D_{\mathrm{flow}}}
\sum_{(s,a)\in\mathcal T}
\left(\frac{d^{\pi_e}(s,a)}{d^{\pi_b}(s,a)}\right)^2x(s,a)
-\Psi(d^{\pi_b})
\right]+O(\alpha^2).
\]
The coefficient is nonnegative because $x=d^{\pi_b}$ is feasible, and it is positive precisely when a feasible occupancy has a larger average squared importance ratio than the baseline. Since occupancies are normalized discounted visitation measures, maximizing the displayed linear objective is equivalent to solving a discounted planning problem with auxiliary reward $(d^{\pi_e}/d^{\pi_b})^2$ on $\mathcal T$ and zero elsewhere. This is a local interpretation of the original design objective under offline coverage; the expansion does not apply when the offline logs leave part of the target support uncovered.

The proofs are in Appendices~\ref{app:exact_decomposition_proof} and~\ref{app:first_order_proof}.

\section{Soft State Aggregation}
\label{app:soft_state_aggregation}

\subsection{Construction of State Cells and Aggregation Weights}

We construct the state cells from the offline trajectories $\mathcal D_{\mathrm{off}}$ and the initial simulated target trajectories $\mathcal S_{\mathrm{init}}=\textsc{Sim}(\pi_e,\mathcal D_{\mathrm{off}})$. We standardize each state feature using its mean and standard deviation estimated from this pool, and write $z(s)$ for the resulting feature vector. We then fit $K$-means to obtain $K_\varphi$ cluster centers $c_1,\ldots,c_{K_\varphi}$.

Let $\mathcal N_m(s)$ index the $m$ cluster centers nearest to $z(s)$. For a bandwidth $\tau>0$, the weight of state $s$ in cell $j$ is
\[
\varphi_j(s)=
\frac{\mathbf 1\{j\in\mathcal N_m(s)\}\exp\{-\|z(s)-c_j\|^2/\tau\}}
{\sum_{\ell\in\mathcal N_m(s)}\exp\{-\|z(s)-c_\ell\|^2/\tau\}},
\qquad j=1,\ldots,K_\varphi.
\]
Thus, each state has positive weights only on its $m$ nearest cells, and the weights sum to one. The bandwidth controls how evenly the weight is distributed among these cells. The feature standardization, cluster centers, and weighting rule remain fixed across audit rounds. In both continuous-state experiments, we use $m=4$ and set $\tau$ to the median squared distance from a standardized state in the clustering pool to its nearest center, excluding distances at most $10^{-12}$. The features and numbers of cells are given in the corresponding numerical studies.

\subsection{Estimation of Cell Transitions and Treatment of Terminal States}
\label{app:cell_transitions}

We index the designated zero-reward terminal state $s_{\mathrm{term}}$ by cell $0$ and let $\mathcal J=\{0,1,\ldots,K_\varphi\}$. Extend the weights by $\varphi_0(s_{\mathrm{term}})=1$, $\varphi_j(s_{\mathrm{term}})=0$ for $j\ge1$, and $\varphi_0(s)=0$ for non-terminal states. Terminated trajectories are padded with $s_{\mathrm{term}}$, zero rewards, and a fixed dummy action when computing the normalized $H$-step empirical occupancy. In the final value estimator, the weighted sum over non-terminal cells is zero at $s_{\mathrm{term}}$, so extending the time sum to $H-1$ leaves the estimate unchanged.

Let $T_i\le H$ be the number of recorded steps in trajectory $i$ before padding. In the episodic evaluations, the abstract transition model sends the last recorded transition to cell $0$, including when the episode reaches $H$. For a trajectory batch $\mathcal D$ and $j,\ell\ge1$, the weighted transition counts are
\[
\begin{aligned}
C_{\mathcal D}(j,a,\ell)
&=\sum_{i\in\mathcal D}\sum_{t=0}^{T_i-2}
\varphi_j(s_t^{(i)})\mathbf 1\{a_t^{(i)}=a\}\varphi_\ell(s_{t+1}^{(i)}),\\
C_{\mathcal D}(j,a,0)
&=\sum_{i\in\mathcal D}\varphi_j(s_{T_i-1}^{(i)})\mathbf 1\{a_{T_i-1}^{(i)}=a\}.
\end{aligned}
\]
We normalize over all destination cells, allowing a pseudo-count $\beta\ge0$:
\[
\widehat P_{\mathcal D}(\ell\mid j,a)
=\frac{C_{\mathcal D}(j,a,\ell)+\beta}
{\sum_{u\in\mathcal J}C_{\mathcal D}(j,a,u)+(K_\varphi+1)\beta},
\qquad j\ge1,\quad \ell\in\mathcal J.
\]
We use $\beta=0$ in both continuous-state experiments. When the denominator is zero, we use the uniform distribution on $\mathcal J$. The terminal row is fixed at $\widehat P_{\mathcal D}(0\mid0,a)=1$, with all other entries zero. At round $k$, we set $\widehat P_k=\widehat P_{\mathcal D_k\cup\mathcal S_k}$. We estimate the initial cell distribution from the real trajectories by
\[
\widehat d_{0,k}(j)=\frac{1}{|\mathcal D_k|}\sum_{i\in\mathcal D_k}\varphi_j(s_0^{(i)}),
\qquad j\in\mathcal J.
\]

\section{Additional Statistical Results}
\label{app:statistical_results}

\subsection{Central Limit Theorem for the Plug-in MIS Estimator}
\label{app:clt_conditions}
\label{app:asymptotics}

We consider the fixed-policy sampling scheme of Section~\ref{sec:var_just}, with the additional assumption that the offline and probe samples are mutually independent. Observed rewards are bounded by $R_{\max}$. Fix $\alpha\in(0,1)$ and let $n\to\infty$ along sample sizes for which $n_{\mathrm{off}}=(1-\alpha)n$ and $B=\alpha n$ are integers. Assume $\mathrm{supp}(d^{\pi_e})\subseteq\mathrm{supp}(q_\alpha^\mu)$ and write $w=d^{\pi_e}/q_\alpha^\mu$, with $w=0$ outside the target support. For a trajectory generated by $\nu\in\{\pi_b,\mu\}$, define the normalized weighted return
\[
Z^{(\nu)}=(1-\gamma)\sum_{t\ge0}\gamma^t
w(s_t^{(\nu)},a_t^{(\nu)})R_t^{(\nu)},
\qquad
m_\nu=\mathbb E_\nu[Z^{(\nu)}].
\]
These variables have finite second moments because the state--action space is finite, $w$ is bounded on the target support, and observed rewards are bounded.

\begin{assumption}[Plug-in conditions for the fixed-policy CLT]
\label{ass:clt}
\leavevmode
\begin{enumerate}
\item The weights $\widehat w$ are fitted using data independent of the evaluation trajectories, as with sample splitting, and
\[
\|\widehat w-w\|_{L^2(q_\alpha^\mu)}\xrightarrow{p}0.
\]
\item The plug-in bias is negligible at the $\sqrt n$ scale:
\[
\sqrt n\left|\sum_{s,a}
\bigl(\widehat w(s,a)-w(s,a)\bigr)r(s,a)q_\alpha^\mu(s,a)\right|
\xrightarrow{p}0.
\]
\end{enumerate}
\end{assumption}

Here $n$ counts the evaluation trajectories; data used to fit the weights are separate. These independence conditions concern this CLT, whereas Theorem~\ref{thm:final-mse} allows the adaptive estimator to reuse its data.

\begin{theorem}[CLT under fixed-policy, fixed-split sampling]
\label{thm:plugin-mis-short}
Under the setup above and Assumption~\ref{ass:clt}, let
\[
\sigma_b^2=\Var(Z^{(b)}),\qquad
\sigma_\mu^2=\Var(Z^{(\mu)}),\qquad
\sigma^2=(1-\alpha)\sigma_b^2+\alpha\sigma_\mu^2.
\]
Then the estimator in~\eqref{eq:MIS} satisfies
\[
\sqrt n(1-\gamma)\bigl\{\widehat J(\widehat w)-J(\pi_e)\bigr\}
\xrightarrow{d}\mathcal N(0,\sigma^2).
\]
\end{theorem}

Appendix~\ref{app:variance_char} bounds $\sigma^2$ above by $R_{\max}^2\Psi(q_\alpha^\mu)$. Thus, the limiting variance of $\sqrt n\{\widehat J(\widehat w)-J(\pi_e)\}$ is at most $R_{\max}^2\Psi(q_\alpha^\mu)/(1-\gamma)^2$.

The proof is in Appendix~\ref{app:clt_proof}.

\subsection{Asymptotic Variance Bound for the Plug-in MIS Estimator}
\label{app:variance_char}

\begin{theorem}[Variance upper bound under fixed-split sampling]
\label{thm:variance-psi}
Consider the fixed-policy sampling scheme of Section~\ref{sec:var_just}, with independent trajectories within each sample and arbitrary dependence between the two samples. Let $\alpha=B/n\in[0,1]$ and suppose that $q_\alpha^\mu$ is positive on the target support. Define $Z^{(\nu)}$ as in Appendix~\ref{app:clt_conditions} and set $\sigma^2=(1-\alpha)\Var(Z^{(b)})+\alpha\Var(Z^{(\mu)})$. If
\[
\mathbb E_\nu[R_t^2\mid s_t=s,a_t=a]\le\overline m_2
\]
for all $t\ge0$, both sampling policies, and all $(s,a)\in\mathrm{supp}(d^{\pi_e})$, then
\[
\sigma^2\le\overline m_2\,\Psi(q_\alpha^\mu),
\qquad
\Var\bigl(\widehat J(w)\bigr)
\le\frac{2\overline m_2}{n(1-\gamma)^2}\Psi(q_\alpha^\mu).
\]
In particular, bounded observed rewards allow $\overline m_2=R_{\max}^2$.
\end{theorem}

The proof is in Appendix~\ref{app:variance_proof}.

\subsection{Plug-in Estimation of Occupancy Ratios}
\label{app:ratio_est}
The simplest ratio estimator is the plug-in $\hat{w} = \hat{d}^{\pi_e}/\hat{q}$, where $\hat{d}^{\pi_e}$ and $\hat{q}$ are separate estimators of $d^{\pi_e}$ and $q$. The following proposition decomposes the resulting estimation error.

\begin{proposition}[Plug-in ratio error decomposition]
\label{prop:plugin_ratio}
Assume $\hat{q}(s,a) \ge c\,q(s,a)$ for all $(s,a)\in\mathrm{supp}(d^{\pi_e})\cup\mathrm{supp}(\hat{d}^{\pi_e})$ and some constant $c\in(0,1]$. Then
\[
\|\hat w - w\|_{L^2(q)}^2
\le
\frac{2}{c^2}
\left[
\int \frac{(\hat{d}^{\pi_e} - d^{\pi_e})^2}{q}
+
\int w^2 \frac{(\hat q - q)^2}{q}
\right].
\]
\end{proposition}

\begin{corollary}[Relative error control]
\label{cor:relative_error}
Under the assumptions of Proposition~\ref{prop:plugin_ratio}, suppose further that $\mathrm{supp}(\hat d^{\pi_e})\subseteq\mathrm{supp}(d^{\pi_e})$ and
\[
|\hat{d}^{\pi_e} - d^{\pi_e}| \le \rho_p \, d^{\pi_e},
\qquad
|\hat q - q| \le \rho_q \, q
\quad \text{on } \mathrm{supp}(d^{\pi_e}).
\]
Then
\[
\|\hat w - w\|_{L^2(q)}^2
\le
\frac{2}{c^2}(\rho_p^2 + \rho_q^2)\, \mathbb E_q[w^2].
\]
\end{corollary}

Under the relative-error conditions of Corollary~\ref{cor:relative_error}, $\Psi(q)=\mathbb E_q[w^2]$ controls an upper bound on ratio-estimation error. With common error controls across designs, reducing $\Psi(q)$ reduces this bound and the oracle variance upper bound in Section~\ref{sec:var_just}.

The proofs are in Appendices~\ref{app:ratio_proof} and~\ref{app:relative_error_proof}.

\begingroup

\subsection{Coordinate Concentration under Sequential Sampling}
\label{app:coordinate_concentration}

The following lemma is used in the proof of Theorem~\ref{thm:final-mse}. Let $\mathcal T=\mathrm{supp}(d_H^{\pi_e})$ and $m=|\mathcal T|$.

\begin{lemma}[Coordinate concentration under sequential sampling]
\label{lem:adaptive_coordinate}
Let $\mathcal H_i$ be a filtration indexed by trajectories. For each $z\in\mathcal T$, suppose $Y_i(z)$ is $\mathcal H_i$-measurable, $|Y_i(z)|\le1$, and
\[
\mathbb E[Y_i(z)^2\mid\mathcal H_{i-1}]\le p_i(z),
\qquad 0\le p_i(z)\le1,
\]
where $p_i(z)$ is $\mathcal H_{i-1}$-measurable. Write
\[
S_n(z)=\sum_{i=1}^n\{Y_i(z)-\mathbb E[Y_i(z)\mid\mathcal H_{i-1}]\},
\qquad q(z)=\frac1n\sum_{i=1}^n p_i(z).
\]
For $\zeta\in(0,1)$ and $\ell=\log\{2\mathrm e\,m[1+\lceil\log_2 n\rceil]/\zeta\}$, with probability at least $1-\zeta$,
\[
\frac{|S_n(z)|}{n}\le 2\sqrt{\frac{q(z)\ell}{n}}+\frac{4\ell}{n}
\qquad\text{for all }z\in\mathcal T.
\]
\end{lemma}

The proof is in Appendix~\ref{app:coordinate_concentration_proof}.
\endgroup

\section{Proofs}
\label{app:sec5}
\label{app:proofs}

\subsection{Proof of Theorem~\ref{thm:flow_reach_gap}: Limits of Probe Design}
\label{app:design_limits_proof}

\begin{proof}
Write $p=d^{\pi_e}$ and $b=d^{\pi_b}$. We interpret each summand $p(s,a)^2/q(s,a)$ as zero when $p(s,a)=0$, and as $+\infty$ when $p(s,a)>0$ and $q(s,a)=0$.

\textit{Bounds and attainment.}
Since $\Psi(q)=1+\chi^2(p\|q)$, every probability distribution $q$ satisfies $\Psi(q)\ge1$, with equality if and only if $q=p$. Choosing $x=p\in\mathcal D_{\mathrm{flow}}$ gives
\[
\Psi\bigl((1-\alpha)b+\alpha p\bigr)
\le\sum_{(s,a):p(s,a)>0}\frac{p(s,a)^2}{\alpha p(s,a)}
=\frac1\alpha.
\]
The set $\mathcal Q_\alpha$ is compact, and the extended-valued objective $\Psi$ is lower semicontinuous. The finite upper bound therefore ensures that its minimum is attained. In particular, $1\le\Psi^\star\le1/\alpha$.

\textit{Characterization of equality.}
If $\Psi^\star=1$, a minimizer must equal $p$. Every $q\in\mathcal Q_\alpha$ satisfies $q\ge(1-\alpha)b$ coordinatewise, so $p\ge(1-\alpha)b$ is necessary.

Conversely, suppose $p\ge(1-\alpha)b$ and set
\[
x^\star=\frac{p-(1-\alpha)b}{\alpha}.
\]
This vector is nonnegative and sums to one. Subtracting $(1-\alpha)$ times the Bellman flow equation for $b$ from the equation for $p$, then dividing by $\alpha$, gives
\[
\sum_a x^\star(s,a)
=(1-\gamma)d_0(s)
+\gamma\sum_{s',a'}x^\star(s',a')P(s\mid s',a')
\quad\text{for every }s.
\]
Hence $x^\star\in\mathcal D_{\mathrm{flow}}$, and $(1-\alpha)b+\alpha x^\star=p$ achieves $\Psi^\star=1$. If the coordinatewise condition fails, no feasible mixture equals $p$. Since the minimum is attained, it is then strictly larger than one.
\end{proof}

\subsection{Proof of Proposition~\ref{prop:kkt_flow}: Optimal Probe Occupancy}
\label{app:kkt_proof}

\begin{proof}
Write $p=d^{\pi_e}$, $b=d^{\pi_b}$, and $q_x=(1-\alpha)b+\alpha x$. Define
\[
F(x)=\Psi(q_x)=\sum_{s,a}\frac{p(s,a)^2}{q_x(s,a)},
\qquad x\in\mathcal D_{\mathrm{flow}},
\]
with zero contribution whenever $p(s,a)=0$. Choosing $x=p$ gives a finite objective, so every minimizer $x^*$ has $q_{x^*}>0$ on the target support. The objective is convex and differentiable near $x^*$, with derivative zero on coordinates where $p=0$. Since the feasible set is a polytope, the KKT conditions apply.

Associate multipliers $\lambda_s$ with the flow constraints and $\eta(s,a)\ge0$ with $x(s,a)\ge0$. Use the Lagrangian
\[
\begin{aligned}
\mathcal L(x,\lambda,\eta)
={}&F(x)+\sum_s\lambda_s\left[
\sum_a x(s,a)-(1-\gamma)d_0(s)
-\gamma\sum_{s',a'}x(s',a')P(s\mid s',a')\right]\\
&-\sum_{s,a}\eta(s,a)x(s,a).
\end{aligned}
\]
Writing $q^*=q_{x^*}$, stationarity at $x^*$ gives
\[
-\alpha\left(\frac{p(s,a)}{q^*(s,a)}\right)^2
+\lambda_s-\gamma\sum_{s'}\lambda_{s'}P(s'\mid s,a)-\eta(s,a)=0,
\]
where the ratio is zero when $p(s,a)=0$. Hence
\begin{equation}
\label{eq:kkt_raw}
\alpha\left(\frac{p(s,a)}{q^*(s,a)}\right)^2
=\lambda_s-\gamma\sum_{s'}\lambda_{s'}P(s'\mid s,a)-\eta(s,a).
\end{equation}

The remaining conditions are $x^*(s,a)\ge0$, $\eta(s,a)\ge0$, and $\eta(s,a)x^*(s,a)=0$, together with flow feasibility. Since $\eta(s,a)\ge0$, equation~\eqref{eq:kkt_raw} gives the stated inequality. If $x^*(s,a)>0$, complementary slackness forces $\eta(s,a)=0$, giving equality. If the inequality is strict, then $\eta(s,a)>0$ and therefore $x^*(s,a)=0$.
\end{proof}

\subsection{Proof of Proposition~\ref{thm:exact_decomp}: Exact Decomposition of Occupancy-Mismatch Reduction}
\label{app:exact_decomposition_proof}

\begin{proof}
We use the notation of Appendix~\ref{app:probing_gain}.
Write $p=d^{\pi_e}$ and $b=d^{\pi_b}$. Since $q_\alpha^\mu=b+\alpha(x-b)$,
\[
\begin{aligned}
\Psi(b)-\Psi(q_\alpha^\mu)
&=\sum_{(s,a)\in\mathcal T}p(s,a)^2
\left(\frac{1}{b(s,a)}-\frac{1}{q_\alpha^\mu(s,a)}\right)\\
&=\alpha\sum_{(s,a)\in\mathcal T}
\frac{p(s,a)^2\bigl(x(s,a)-b(s,a)\bigr)}
{b(s,a)\,q_\alpha^\mu(s,a)}.
\end{aligned}
\]
\end{proof}

\subsection{Proof of Proposition~\ref{thm:first_order}: Small-Budget Expansion of Occupancy-Mismatch Reduction}
\label{app:first_order_proof}

\begin{proof}
We use the notation of Appendix~\ref{app:probing_gain}.
Write $p=d^{\pi_e}$ and $b=d^{\pi_b}$, and let $h=x-b$ and $b_{\min}=\min_{(s,a)\in\mathcal T}b(s,a)>0$. Coordinatewise on $\mathcal T$,
\[
\frac{1}{b+\alpha h}
=\frac1b-\frac{\alpha h}{b^2}
+\frac{\alpha^2h^2}{b^2(b+\alpha h)}.
\]
Multiplying by $p^2$ and summing gives
\[
\Psi(b)-\Psi(b+\alpha h)
=\alpha\sum_{\mathcal T}\frac{p^2h}{b^2}
-\alpha^2\sum_{\mathcal T}\frac{p^2h^2}{b^2(b+\alpha h)}.
\]
For $0\le\alpha\le1/2$, we have $b+\alpha h=(1-\alpha)b+\alpha x\ge b/2$. Since $|h|\le1$ and $\sum_{\mathcal T}p^2\le1$, the absolute remainder is at most $2\alpha^2/b_{\min}^3$, independently of $x$. This proves the uniform expansion.
\end{proof}

\subsection{Proof of Theorem~\ref{thm:mse_decomp}: MSE Bound for the Plug-in MIS Estimator}
\label{app:fixed_mse_proof}

\begin{proof}
Write $q=q_\alpha^\mu$, $w=d^{\pi_e}/q$, and let $\widehat J(w)$ denote the oracle version of~\eqref{eq:MIS}. For a function $g$ on the state--action space, define the empirical discounted average over the two fixed sampling groups by
\[
\mathbb E_n[g]=\frac1n\sum_{i=1}^n\sum_{t\ge0}(1-\gamma)\gamma^t g(s_t^{(i)},a_t^{(i)}).
\]
For deterministic $g$, $\mathbb E[\mathbb E_n[g]]=\mathbb E_q[g]$. The observed-reward estimator is
\[
\widehat J(v)=\frac1n\sum_{i=1}^n\sum_{t\ge0}\gamma^t v(s_t^{(i)},a_t^{(i)})R_t^{(i)}.
\]
Decompose its error as
\[
\widehat J(\widehat w)-J(\pi_e)
=\underbrace{\widehat J(w)-J(\pi_e)}_{A}
+\underbrace{\widehat J(\widehat w)-\widehat J(w)}_{D}.
\]

\textit{Oracle error.}
The mixture identity and the MDP mean reward law give
\[
\mathbb E[\widehat J(w)]
=\frac1{1-\gamma}\sum_{s,a}q(s,a)w(s,a)r(s,a)
=J(\pi_e).
\]
By Theorem~\ref{thm:variance-psi}, whose oracle variance bound requires independence only within each sampling group,
\[
\mathbb E[A^2]=\Var\bigl(\widehat J(w)\bigr)
\le\frac{2R_{\max}^2}{n(1-\gamma)^2}\Psi(q).
\]

\textit{Weight-estimation error.}
On the target support, write $\hat d^{\pi_e}=d^{\pi_e}(1+u)$ and $\hat q=q(1+v)$. The assumptions imply $|u|\le\rho_p$, $|v|\le\rho_q$, and $1+v\ge c$, almost surely. Hence
\[
|\widehat w-w|=w\left|\frac{u-v}{1+v}\right|
\le\frac{\rho_p+\rho_q}{c}\,w.
\]
Both weights vanish outside the target support, so the inequality holds everywhere. Bounded observed rewards give the pathwise bound
\[
|D|\le\frac{R_{\max}}{1-\gamma}\frac{\rho_p+\rho_q}{c}\,\mathbb E_n[w].
\]
The empirical weights defining $\mathbb E_n$ sum to one. Jensen's inequality and $\mathbb E[\mathbb E_n[w^2]]=\Psi(q)$ therefore yield
\[
\mathbb E[D^2]
\le\frac{R_{\max}^2(\rho_p+\rho_q)^2}{c^2(1-\gamma)^2}\Psi(q)
\le\frac{2C_0R_{\max}^2}{c^2n(1-\gamma)^2}\Psi(q).
\]
No independence between $\widehat w$ and the evaluation data is needed for this pointwise argument. Finally, $(A+D)^2\le2A^2+2D^2$ gives
\[
\mathbb E\!\left[\bigl(\widehat J(\widehat w)-J(\pi_e)\bigr)^2\right]
\le\frac{4R_{\max}^2}{n(1-\gamma)^2}
\left(1+\frac{C_0}{c^2}\right)\Psi(q),
\]
which proves the claimed bound.
If a designated terminal state has identically zero subsequent rewards, set $w^{-}(s,a)=w(s,a)\mathbf1\{s\in\mathcal S^{-}\}$ and $\widehat w^{-}(s,a)=\widehat w(s,a)\mathbf1\{s\in\mathcal S^{-}\}$. This leaves both estimators unchanged. The oracle variance argument and the weight-estimation argument above apply to these restricted weights, with $\mathbb E_q[(w^{-})^2]=\sum_{s\in\mathcal S^{-},\,a\in\mathcal A}d^{\pi_e}(s,a)^2/q(s,a)$ in place of $\Psi(q)$. Thus the same constant gives the stated extension to zero-reward terminal states.
\end{proof}

\subsection{Proof of Theorem~\ref{thm:plugin-mis-short}: Central Limit Theorem for the Plug-in MIS Estimator}
\label{app:clt_proof}

\begin{proof}
Write $q=q_\alpha^\mu$ and let $\widehat J(w)$ be the oracle estimator.

\textit{Oracle fluctuation.}
The two groups have means $m_b=\sum_{s,a}d^{\pi_b}(s,a)w(s,a)r(s,a)$ and $m_\mu=\sum_{s,a}d^\mu(s,a)w(s,a)r(s,a)$. Their weighted average satisfies
\[
(1-\alpha)m_b+\alpha m_\mu
=\sum_{s,a}q(s,a)w(s,a)r(s,a)
=(1-\gamma)J(\pi_e).
\]
Consequently,
\[
\begin{aligned}
\sqrt n(1-\gamma)\bigl\{\widehat J(w)-J(\pi_e)\bigr\}
=\frac{1}{\sqrt n}\left[
\sum_{i=1}^{n_{\mathrm{off}}}(Z_i^{(b)}-m_b)
+\sum_{j=1}^{B}(Z_j^{(\mu)}-m_\mu)
\right].
\end{aligned}
\]
Each sum is centered at its own group's mean. The classical CLT applied to the two independent groups gives convergence to $\mathcal N(0,(1-\alpha)\sigma_b^2+\alpha\sigma_\mu^2)$, including a degenerate limit when $\sigma^2=0$.

\textit{Plug-in remainder.}
Let
\[
D_n=(1-\gamma)\bigl\{\widehat J(\widehat w)-\widehat J(w)\bigr\},
\qquad
U^{(\nu)}=(1-\gamma)\sum_{t\ge0}\gamma^t
(\widehat w-w)(s_t^{(\nu)},a_t^{(\nu)})R_t^{(\nu)}.
\]
Conditional on $\widehat w$, the evaluation trajectories are independent with their original policy laws. Therefore,
\[
b_n:=\mathbb E[D_n\mid\widehat w]
=\sum_{s,a}(\widehat w-w)(s,a)r(s,a)q(s,a),
\qquad
\sqrt n\,b_n\xrightarrow{p}0.
\]
Jensen's inequality with weights $(1-\gamma)\gamma^t$ and bounded observed rewards gives
\[
\mathbb E[(U^{(\nu)})^2\mid\widehat w]
\le R_{\max}^2\|\widehat w-w\|_{L^2(d^\nu)}^2.
\]
Combining the two groups yields
\[
\begin{aligned}
n\Var(D_n\mid\widehat w)
&\le R_{\max}^2\left[(1-\alpha)\|\widehat w-w\|_{L^2(d^{\pi_b})}^2
+\alpha\|\widehat w-w\|_{L^2(d^\mu)}^2\right]\\
&=R_{\max}^2\|\widehat w-w\|_{L^2(q)}^2\xrightarrow{p}0.
\end{aligned}
\]
For any $u>0$, conditional Chebyshev gives
\[
\Pr\{\sqrt n|D_n-b_n|>u\}
\le\mathbb E\!\left[\min\left\{1,
\frac{R_{\max}^2\|\widehat w-w\|_{L^2(q)}^2}{u^2}\right\}\right]
\longrightarrow0,
\]
since the integrand is bounded and converges to zero in probability. Thus $\sqrt nD_n\xrightarrow{p}0$. Combining this with the oracle CLT and applying Slutsky's theorem proves the result.
\end{proof}

\subsection{Proof of Theorem~\ref{thm:variance-psi}: Variance Upper Bound}
\label{app:variance_proof}

\begin{proof}
Let $\beta_t=(1-\gamma)\gamma^t$, so $\sum_{t\ge0}\beta_t=1$. Jensen's inequality gives the pathwise bound
\[
\bigl(Z^{(\nu)}\bigr)^2
=\left(\sum_{t\ge0}\beta_t w(s_t,a_t)R_t\right)^2
\le\sum_{t\ge0}\beta_t w(s_t,a_t)^2R_t^2.
\]
Taking expectations and using the conditional second-moment bound,
\[
\mathbb E_\nu[(Z^{(\nu)})^2]
\le\overline m_2\sum_{s,a}d^\nu(s,a)w(s,a)^2.
\]
Since $\Var(Z^{(\nu)})\le\mathbb E_\nu[(Z^{(\nu)})^2]$, the mixture identity yields
\[
\begin{aligned}
\sigma^2
&\le\overline m_2\sum_{s,a}
\bigl[(1-\alpha)d^{\pi_b}(s,a)+\alpha d^\mu(s,a)\bigr]w(s,a)^2\\
&=\overline m_2\sum_{s,a}q_\alpha^\mu(s,a)w(s,a)^2
=\overline m_2\Psi(q_\alpha^\mu).
\end{aligned}
\]
This argument allows dependence across time within a trajectory. Write
\[
S_b=\sum_{i=1}^{n_{\mathrm{off}}}(Z_i^{(b)}-m_b),
\qquad
S_\mu=\sum_{j=1}^{B}(Z_j^{(\mu)}-m_\mu).
\]
Independence within each sample gives $\Var(S_b)=n_{\mathrm{off}}\sigma_b^2$ and $\Var(S_\mu)=B\sigma_\mu^2$. For arbitrary dependence between the samples,
\[
\begin{aligned}
\Var\bigl(\widehat J(w)\bigr)
&=\frac{\mathbb E[(S_b+S_\mu)^2]}{n^2(1-\gamma)^2}\\
&\le\frac{2\mathbb E[S_b^2]+2\mathbb E[S_\mu^2]}{n^2(1-\gamma)^2}
=\frac{2\sigma^2}{n(1-\gamma)^2},
\end{aligned}
\]
which proves the second bound. Under the sampling conditions of Appendix~\ref{app:clt_conditions} and Assumption~\ref{ass:clt}, $\sigma^2$ is the limiting variance of the normalized plug-in estimator.
\end{proof}

\subsection{Proof of Proposition~\ref{prop:plugin_ratio}: Occupancy-Ratio Error Decomposition}
\label{app:ratio_proof}

\begin{proof}
We decompose
\[
\hat w - w
=
\frac{\hat{d}^{\pi_e}}{\hat q} - \frac{d^{\pi_e}}{q}
=
\frac{\hat{d}^{\pi_e} - d^{\pi_e}}{\hat q}
+
d^{\pi_e}\left(\frac{1}{\hat q} - \frac{1}{q}\right)
=
\frac{\hat{d}^{\pi_e} - d^{\pi_e}}{\hat q}
+
\frac{d^{\pi_e}(q - \hat q)}{q \hat q}.
\]
Using $(a+b)^2 \le 2a^2 + 2b^2$, we obtain
\[
(\hat w - w)^2
\le
2 \frac{(\hat{d}^{\pi_e} - d^{\pi_e})^2}{\hat q^2}
+
2 (d^{\pi_e})^2 \frac{(\hat q - q)^2}{q^2 \hat q^2}.
\]
Multiplying by $q$ and integrating yields
\[
\|\hat w - w\|_{L^2(q)}^2
\le
2 \int \frac{(\hat{d}^{\pi_e} - d^{\pi_e})^2 q}{\hat q^2}
+
2 \int (d^{\pi_e})^2 \frac{(\hat q - q)^2}{q \hat q^2}.
\]
Using $\hat q \ge c q$, we have
\[
\frac{q}{\hat q^2} \le \frac{1}{c^2 q},
\qquad
\frac{1}{q \hat q^2} \le \frac{1}{c^2 q^3},
\]
which gives
\[
\|\hat w - w\|_{L^2(q)}^2
\le
\frac{2}{c^2}
\left[
\int \frac{(\hat{d}^{\pi_e} - d^{\pi_e})^2}{q}
+
\int \frac{(d^{\pi_e})^2}{q^3}(\hat q - q)^2
\right].
\]
Finally, noting that $(d^{\pi_e})^2/q^2 = w^2$ yields the result.
\end{proof}

\subsection{Proof of Corollary~\ref{cor:relative_error}: Relative Error Control for Occupancy Ratios}
\label{app:relative_error_proof}

\begin{proof}
The assumptions $|\hat{d}^{\pi_e}-d^{\pi_e}|\le\rho_p d^{\pi_e}$ and $|\hat{q}-q|\le\rho_q q$ pointwise on $\mathrm{supp}(d^{\pi_e})$ imply $(\hat{d}^{\pi_e}-d^{\pi_e})^2\le\rho_p^2 (d^{\pi_e})^2$ and $(\hat{q}-q)^2\le\rho_q^2 q^2$ respectively.
The support condition ensures that both $\hat d^{\pi_e}$ and $d^{\pi_e}$ vanish outside $\mathrm{supp}(d^{\pi_e})$, so the first inequality also holds there. We therefore bound each term in Proposition~\ref{prop:plugin_ratio}:
\[
\int \frac{(\hat{d}^{\pi_e} - d^{\pi_e})^2}{q}
\le
\rho_p^2 \int \frac{(d^{\pi_e})^2}{q}
=
\rho_p^2 \mathbb E_q[w^2],
\]
and
\[
\int w^2 \frac{(\hat q - q)^2}{q}
\le
\rho_q^2 \int w^2 q
=
\rho_q^2 \mathbb E_q[w^2].
\]
Combining the bounds gives the result.
\end{proof}

\begingroup

\subsection{Proof of Lemma~\ref{lem:adaptive_coordinate}: Coordinate Concentration under Sequential Sampling}
\label{app:coordinate_concentration_proof}

\begin{proof}
We use the notation of Appendix~\ref{app:coordinate_concentration}.
For a fixed $z$, the centered increments have absolute value at most $2$, and their predictable quadratic variation $V_n(z)$ is at most $D_n(z):=\sum_i p_i(z)=nq(z)$. The scalar Freedman inequality, in the form stated in Theorem~1.1 of \citet{tropp2011freedman}, gives
\[
\Pr\left\{|S_n(z)|\ge\sqrt{2v\ell}+\frac{4\ell}{3},\ V_n(z)\le v\right\}\le2e^{-\ell}
\]
for deterministic $v>0$. Indeed, the displayed threshold is at least the positive solution of $t^2=2\ell(v+2t/3)$.

Partition the possible values of $D_n(z)\in[0,n]$ into $[0,1]$ and $(2^{j-1},2^j]$, $j=1,\ldots,\lceil\log_2n\rceil$. Apply the inequality using the upper endpoint $v$ of each interval. In the interval containing $D_n(z)$, $V_n(z)\le D_n(z)\le v\le2[D_n(z)+1]$. A union bound over intervals and the $m$ coordinates has failure probability at most $2m[1+\lceil\log_2n\rceil]e^{-\ell}\le\zeta$. On its complement,
\[
|S_n(z)|\le2\sqrt{[D_n(z)+1]\ell}+4\ell/3
\le2\sqrt{D_n(z)\ell}+4\ell,
\]
where the last inequality uses $\ell\ge1$. Division by $n$ proves the result.
\end{proof}
\endgroup

\begingroup

\subsection{Proof of Theorem~\ref{thm:final-mse}: Adaptive Policy-Evaluation Error Bound}
\label{app:adaptive_mse_proof}

\begin{proof}
Let $\mathcal T=\mathrm{supp}(d_H^{\pi_e})$ and $m=|\mathcal T|$. Write $d=d_H^{\pi_e}$ and $L_H=(1-\gamma^H)/(1-\gamma)$ for this proof. We may assume $R_{\max}>0$, since the claim is immediate otherwise. All trajectories below contain $H$ steps under the padding convention of Section~\ref{sec:background}.

\textit{Sequential trajectory contributions.}
Reveal the trajectories in collection order. Let $\mathcal H_{i-1}$ contain the trajectories already revealed and the randomization used to select the policy for trajectory $i$. The sampling conditions of Theorem~\ref{thm:final-mse} ensure that, conditional on this information, trajectory $i$ follows its selected policy. This also holds when trajectories within a round are revealed one at a time. For $z=(s,a)$, define
\[
\begin{aligned}
Z_i(z)&=\frac1{L_H}\sum_{t=0}^{H-1}\gamma^t
\mathbf1\{(s_t^{(i)},a_t^{(i)})=z\},\\
Y_i(z)&=\frac1{R_{\max}L_H}\sum_{t=0}^{H-1}\gamma^t
\mathbf1\{(s_t^{(i)},a_t^{(i)})=z\}R_t^{(i)}.
\end{aligned}
\]
Each $Z_i$ is a probability vector, $\|Y_i\|_1\le1$, and $Y_i(z)^2\le Z_i(z)\le1$. Let $p_i$ be the population $H$-step occupancy of the policy for trajectory $i$. The MDP reward law gives
\[
\begin{gathered}
\mathbb E[Z_i(z)\mid\mathcal H_{i-1}]=p_i(z),\qquad
\mathbb E[Y_i(z)\mid\mathcal H_{i-1}]=p_i(z)r(z)/R_{\max},\\
\mathbb E[Y_i(z)^2\mid\mathcal H_{i-1}]\le p_i(z).
\end{gathered}
\]
Here $p_i=d_H^{\pi_b}$ for offline trajectories and $p_i=d_H^{\mu_k}$ in round $k$, so $n^{-1}\sum_i p_i=q_{H,K}$ pathwise. Lemma~\ref{lem:adaptive_coordinate} therefore applies.

\textit{Final random weights.}
Define $w^\lambda(z)=d(z)/[q_{H,K}(z)+\lambda]$ and
\[
\mathcal J_H(v)=\frac{R_{\max}L_H}{n}\sum_{i=1}^n\sum_z v(z)Y_i(z).
\]
The estimator in~\eqref{eq:mis_final} is exactly $\mathcal J_H(\widehat w_K)$. Its oracle sampling fluctuation is
\[
\begin{aligned}
V&:=\mathcal J_H(w^\lambda)
-L_H\sum_z\frac{d(z)q_{H,K}(z)}{q_{H,K}(z)+\lambda}r(z)\\
&=R_{\max}L_H\sum_{z\in\mathcal T}
\frac{d(z)}{q_{H,K}(z)+\lambda}\frac{S_n(z)}n.
\end{aligned}
\]
The subtracted quantity is a random population functional, not a conditional expectation given all selected policies. On the simultaneous event $E$ of Lemma~\ref{lem:adaptive_coordinate},
\[
|V|\le R_{\max}L_H\left[
2\sqrt{\frac\ell n}\sum_{z\in\mathcal T}
\frac{d(z)\sqrt{q_{H,K}(z)}}{q_{H,K}(z)+\lambda}
+\frac{4\ell}{n}\sum_{z\in\mathcal T}\frac{d(z)}{q_{H,K}(z)+\lambda}
\right].
\]
Cauchy--Schwarz and $\sum_zd(z)=1$ give
\[
\begin{aligned}
\left(\sum_{z\in\mathcal T}
\frac{d(z)\sqrt{q_{H,K}(z)}}{q_{H,K}(z)+\lambda}\right)^2
&\le m\Psi_{H,\lambda}(q_{H,K}),\\
\sum_{z\in\mathcal T}\frac{d(z)}{q_{H,K}(z)+\lambda}&\le1/\lambda.
\end{aligned}
\]
Hence, on $E$,
\[
V^2\le R_{\max}^2L_H^2\left[
\frac{8m\ell}{n}\Psi_{H,\lambda}(q_{H,K})
+\frac{32\ell^2}{n^2\lambda^2}\right].
\]
Concentration is established simultaneously for all target-support coordinates before applying the final data-dependent weights.

\textit{Estimated weights.}
Let $a_K=\delta_{H,K}/\lambda+\tau_{H,K}/\lambda^2$. Since both denominators are at least $\lambda$,
\[
\begin{aligned}
\widehat w_K-w^\lambda
&=\frac{\widehat d^{\pi_e}_{H,K}-d}{\widehat q_K+\lambda}
+\frac{d(q_{H,K}-\widehat q_K)}{(\widehat q_K+\lambda)(q_{H,K}+\lambda)},\\
\|\widehat w_K-w^\lambda\|_\infty&\le a_K.
\end{aligned}
\]
Therefore, using $\|Y_i\|_1\le1$,
\[
|\mathcal J_H(\widehat w_K)-\mathcal J_H(w^\lambda)|
\le R_{\max}L_H a_K.
\]
This pathwise bound requires no independence between the occupancy estimates and the evaluation data. Both estimates and their population counterparts use the same $H$-step normalization.

\textit{Regularization.}
The identity $J_H(\pi_e)=L_H\sum_zd(z)r(z)$ implies
\[
L_H\sum_z\frac{d(z)q_{H,K}(z)}{q_{H,K}(z)+\lambda}r(z)-J_H(\pi_e)
=\mathrm{Bias}_{H,\lambda}(q_{H,K}).
\]
This is a random regularization offset and need not equal the unconditional bias of the estimator.

\textit{Mean squared error.}
On $E$, the preceding decomposition and $(u+v+w)^2\le3(u^2+v^2+w^2)$ give
\[
\begin{aligned}
\bigl(\widehat J(\widehat w_K)-J_H(\pi_e)\bigr)^2
\le{}&R_{\max}^2L_H^2\left[
\frac{24m\ell}{n}\Psi_{H,\lambda}(q_{H,K})
+\frac{96\ell^2}{n^2\lambda^2}+3a_K^2\right]\\
&+3\mathrm{Bias}_{H,\lambda}(q_{H,K})^2.
\end{aligned}
\]
On $E^c$, $|\widehat J(\widehat w_K)|\le R_{\max}L_H/\lambda$ and $|J_H(\pi_e)|\le R_{\max}L_H$. Since $\Pr(E^c)\le\zeta$, this event contributes at most $\zeta R_{\max}^2L_H^2(1+1/\lambda)^2$ to the MSE. Integrating the bound on $E$ and dropping indicators from nonnegative terms proves the theorem.
\end{proof}
\endgroup

\begingroup

\subsection{Proof of Corollary~\ref{cor:adaptive_consistency}: Mean-Square Consistency of Adaptive Policy Evaluation}
\label{app:adaptive_consistency_proof}

\begin{proof}
Take $\zeta=n^{-3}$ for $n\ge2$. Since the MDP and $H$ are fixed, $m$ is fixed and $\ell=O(\log n)$. Because $d_H^{\pi_e}$ is a probability distribution,
\[
\Psi_{H,\lambda_n}(q_{H,K})
\le\frac1{\lambda_n}\sum_{s,a}d_H^{\pi_e}(s,a)^2
\le\frac1{\lambda_n}.
\]
Consequently,
\[
\frac{m\ell}{n}\mathbb E[\Psi_{H,\lambda_n}(q_{H,K})]
\le\frac{m\ell}{n\lambda_n}\longrightarrow0.
\]
Writing $a_K:=\delta_{H,K}/\lambda_n+\tau_{H,K}/\lambda_n^2$, the occupancy-estimation assumptions give
\[
\mathbb E[a_K^2]
\le2\frac{\mathbb E[\delta_{H,K}^2]}{\lambda_n^2}
+2\frac{\mathbb E[\tau_{H,K}^2]}{\lambda_n^4}
\longrightarrow0.
\]
The condition $n\lambda_n/\log n\to\infty$ implies $\ell^2/(n^2\lambda_n^2)\to0$ and $n^{-3}(1+1/\lambda_n)^2\to0$.

It remains to control the regularization offset. On $\{c_{H,K}\ge c\}$,
\[
|\mathrm{Bias}_{H,\lambda_n}(q_{H,K})|
\le R_{\max}\frac{1-\gamma^H}{1-\gamma}\frac{\lambda_n}{c+\lambda_n},
\]
and its absolute value is always at most $R_{\max}(1-\gamma^H)/(1-\gamma)$. Therefore,
\[
\mathbb E[\mathrm{Bias}_{H,\lambda_n}(q_{H,K})^2]
\le R_{\max}^2\left(\frac{1-\gamma^H}{1-\gamma}\right)^2
\left[\left(\frac{\lambda_n}{c+\lambda_n}\right)^2+\Pr(c_{H,K}<c)\right]
\longrightarrow0.
\]
Every term in Theorem~\ref{thm:final-mse} now tends to zero, proving mean-square consistency.
\end{proof}
\endgroup

\begingroup

\subsection{Proof of Corollary~\ref{cor:continuous_extension}: Continuous-State Policy-Evaluation Error Bound}
\label{app:continuous_extension_proof}

\begin{proof}
Fix $\varphi$ independently of the trajectories used for estimation, as required by Corollary~\ref{cor:continuous_extension}. Let $L_H=(1-\gamma^H)/(1-\gamma)$, let $\mathcal T=\{(j,a):1\le j\le K_\varphi,\ d_H^{\pi_e,\varphi}(j,a)>0\}$, and suppose $m=|\mathcal T|\ge1$.

\textit{Sequential trajectory contributions.}
For trajectory $i$ and a non-terminal cell--action pair $z=(j,a)$, define
\[
\begin{aligned}
Z_i(z)&=\frac1{L_H}\sum_{t=0}^{H-1}\gamma^t
\varphi_j(s_t^{(i)})\mathbf1\{a_t^{(i)}=a\},\\
Y_i(z)&=\frac1{R_{\max}L_H}\sum_{t=0}^{H-1}\gamma^t
\varphi_j(s_t^{(i)})\mathbf1\{a_t^{(i)}=a\}R_t^{(i)}.
\end{aligned}
\]
The case $R_{\max}=0$ is immediate. Nonnegative aggregation weights and the zero-reward terminal convention give $\sum_zZ_i(z)\le1$, $\|Y_i\|_1\le1$, and $Y_i(z)^2\le Z_i(z)\le1$. Reveal trajectories in collection order using the filtration from the proof of Theorem~\ref{thm:final-mse}. If $\pi_i$ is the policy for trajectory $i$, write $p_i(z)=d_H^{\pi_i,\varphi}(j,a)$. The conditional moments satisfy
\[
\begin{aligned}
\mathbb E[Z_i(z)\mid\mathcal H_{i-1}]&=p_i(z),\\
\mathbb E[Y_i(z)\mid\mathcal H_{i-1}]
&=\frac{p_i(z)\bar r_H^{\pi_i,\varphi}(j,a)}{R_{\max}},\\
\mathbb E[Y_i(z)^2\mid\mathcal H_{i-1}]&\le p_i(z).
\end{aligned}
\]
When $p_i(z)=0$, the second identity holds with both sides zero. Also, $n^{-1}\sum_i p_i=q_{H,K}^{\varphi}$ pathwise. Lemma~\ref{lem:adaptive_coordinate} therefore yields the same coordinate concentration bound on $\mathcal T$.

\textit{Sampling and estimated-weight errors.}
The target and estimated non-terminal cell occupancies each have mass at most one. Thus the bounds for the final random weights and estimated weights in the proof of Theorem~\ref{thm:final-mse} remain valid. The estimator in~\eqref{eq:continuous_mis_final} is exactly
\[
\widehat J(\widehat w_K^{\varphi})
=\frac{R_{\max}L_H}{n}\sum_{i=1}^n\sum_z\widehat w_K^{\varphi}(z)Y_i(z).
\]
\textit{Aggregation error.}
Because non-terminal aggregation weights sum to one and terminal rewards are zero,
\[
J_H(\pi_e)
=L_H\sum_{j=1}^{K_\varphi}\sum_{a\in\mathcal A}
d_H^{\pi_e,\varphi}(j,a)\bar r_H^{\pi_e,\varphi}(j,a).
\]
Write $d(z)=d_H^{\pi_e,\varphi}(j,a)$ and $w_\lambda(z)=d(z)/(q_{H,K}^{\varphi}(z)+\lambda)$. The difference between the predictable reward contribution and its counterpart using the target cell means is
\[
\Delta_{H,K}
=\frac{L_H}{n}\sum_{i=1}^n\sum_{z\in\mathcal T}
w_\lambda(z)p_i(z)
\bigl\{\bar r_H^{\pi_i,\varphi}(j,a)-\bar r_H^{\pi_e,\varphi}(j,a)\bigr\}.
\]
The assumed reward bound and $n^{-1}\sum_i p_i=q_{H,K}^{\varphi}$ imply, pathwise,
\[
|\Delta_{H,K}|
\le L_H\eta\sum_{z\in\mathcal T}
\frac{d(z)q_{H,K}^{\varphi}(z)}{q_{H,K}^{\varphi}(z)+\lambda}
\le L_H\eta.
\]
\textit{Mean squared error.}
Centering $Y_i$ at its own conditional mean gives the same concentration and estimated-weight bounds as in the proof of Theorem~\ref{thm:final-mse}. The predictable reward contribution minus $J_H(\pi_e)$ is now $\mathrm{Bias}_{H,\lambda}^{\varphi}(q_{H,K}^{\varphi})+\Delta_{H,K}$. Applying the same three-term squared-error inequality on the concentration event gives the bound with its last term replaced by $3\{\mathrm{Bias}_{H,\lambda}^{\varphi}(q_{H,K}^{\varphi})+\Delta_{H,K}\}^2$. Using
\[
3\{\mathrm{Bias}_{H,\lambda}^{\varphi}(q_{H,K}^{\varphi})+\Delta_{H,K}\}^2
\le6\mathrm{Bias}_{H,\lambda}^{\varphi}(q_{H,K}^{\varphi})^2+6L_H^2\eta^2
\]
proves the stated bound. The contribution from the complement of the concentration event is unchanged because it uses only bounded rewards and bounded estimated weights. When $\eta=0$, $\Delta_{H,K}=0$, so the original constants apply. The argument uses the original process and does not require the cell process to be Markov.

\textit{Consistency.}
The fixed aggregation and finite action space give a fixed finite $m$. The inequality $\Psi_{H,\lambda_n}^{\varphi}(q_{H,K}^{\varphi})\le1/\lambda_n$ continues to hold because the target cell occupancies have total mass at most one. Also $|\bar r_H^{\pi_e,\varphi}(j,a)|\le R_{\max}$, so the final-coverage condition controls the regularization term exactly as in the proof of Corollary~\ref{cor:adaptive_consistency}. The sampling and occupancy-estimation terms vanish under the stated rate conditions, leaving a limiting upper bound of $6L_H^2\eta^2$. If $\eta=\eta_n\to0$, this term also vanishes, proving mean-square consistency.
\end{proof}
\endgroup

\end{document}